\documentclass[onecolumn,draft,11pt]{IEEEtran}
\ifCLASSINFOpdf
\else
\fi
\newcommand{\GRS}{{\mathrm{GRS}}}
\newcommand{\Hull}{{\mathrm{Hull}}}
\newcommand{\rank}{{\mathrm{rank}}}
\newcommand{\diag}{{\mathrm{diag}}}
\newcommand{\C}{{\mathcal{C}}}
\newcommand{\D}{{\mathcal{D}}}
\newcommand{\F}{{\mathbb{F}}}
\newcommand{\bc}{{{\bf c}}}
\newcommand{\wt}{{{\bf wt}}}

\usepackage{amsmath}
\usepackage{amssymb}
\usepackage{bm}
\usepackage{amssymb}
\usepackage{amsthm}
\usepackage{multirow,booktabs}
\usepackage{cases}
\usepackage{tabularx}
\usepackage{adjustbox}
\usepackage[figuresright]{rotating}
\usepackage{longtable}
\usepackage{booktabs}
\usepackage{multirow}
\usepackage{color}
\usepackage{cite}
\usepackage{multirow,booktabs}
\usepackage{cases}
\usepackage{tabularx}
\usepackage{adjustbox}
\usepackage[figuresright]{rotating}
\usepackage{longtable}
\usepackage{booktabs}
\usepackage{multirow}
\usepackage{color}
\usepackage{ulem}
\newtheorem{theorem}{Theorem}
\newtheorem{proposition}[theorem]{Proposition}
\newtheorem{remark}{Remark}
\newtheorem{definition}[theorem]{Definition}
\newtheorem{lemma}[theorem]{Lemma}
\newtheorem{corollary}[theorem]{Corollary}
\newtheorem{example}[theorem]{Example}

\newcommand{\tabincell}[2]{\begin{tabular}{@{}#1@{}}#2\end{tabular}}

\begin{document}
\title{The hull of two classical propagation rules and their applications \thanks{First and second authors were supported by the National Natural Science Foundation of China 
		(Nos.U21A20428 and 12171134). Third author was supported by Grant TED2021-130358B-I00 funded by MCIN/AEI/10.13039/501100011033 and by the “European Union NextGenerationEU/PRTR”}}

\author{Yang~Li, Shixin~Zhu$^\dag$\thanks{$^\dag$ Corresponding author}
        and~Edgar~Mart\'inez-Moro% <-this % stops a space
\thanks{Yang~Li and  Shixin~Zhu are with the School of Mathematics, Hefei University of Technology, Hefei 230601, China (e-mail: yanglimath@163.com, zhushixinmath@hfut.edu.cn).}% <-this % stops a space
\thanks{Edgar~Mart\'inez-Moro is with the Institute of Mathematics University of Valladolid, Spain (e-mail: Edgar.Martinez@uva.es).}% <-this % stops a space
\thanks{Manuscript received --; revised --}}


\markboth{Journal of \LaTeX\ Class Files}%
{Shell \MakeLowercase{\textit{et al.}}: Bare Demo of IEEEtran.cls for Journals}



\maketitle

\begin{abstract}
	In this work, we study and 
determine the dimensions of Euclidean and Hermitian hulls of two classical 
propagation rules, namely, the direct sum construction and the $(\mathbf{u},\mathbf{u+v})$-
construction. Some new criteria for the resulting codes derived from these two propagation 
rules being self-dual, self-orthogonal, or linear complementary dual (LCD) codes are given. 
As  an application, we construct some linear codes with prescribed hull dimensions, many new 
binary, ternary Euclidean formally self-dual (FSD) LCD codes, and quaternary Hermitian FSD 
LCD codes. 
Some new even-like, odd-like, Euclidean and Hermitian self-orthogonal codes 
are also obtained. Many of {these} codes are also 
(almost) optimal according to the Database maintained by Markus Grassl. %\cite{Grass}. 
Our methods contribute positively to improve the lower bounds on the minimum distance of known LCD codes. 
\end{abstract}


% \fi
\IEEEpeerreviewmaketitle

\section{Introduction}\label{sec1}
Let $q=p^h$ be a prime power. We denote by  $\F_q$  the finite field of order $q$ and by $\F_q^n$ 
 the  $\F_q$-vector space  provided with the Hamming distance. 
We call a $k$-dimensional subspace of $\F_q^n$ with minimum Hamming distance between two different vectors $d$ 
a $q$-ary linear code $\C$, and we will denote it as an $[n,k,d]_q$-code. 
Throughout this paper, if $\C$ has the largest minimum distance among all $[n,k]_q$ linear codes, 
then $\C$ is called \emph{optimal} and if its minumum distance is $d$ and there exists an optimal $[n,k,d+1]_q$ linear code, 
then the code $\C$ is called \emph{almost optimal}.

For any two vectors $\mathbf{x}=(x_1,x_2,\dots,x_n)$, $\mathbf{y}=(y_1,y_2,\dots,y_n)\in \F_q^n$, 
the Euclidean (resp. Hermitian) inner product of $\mathbf{x}$ and $\mathbf{y}$ is defined as 
$(\mathbf{x},\mathbf{y})_E=\sum_{i=1}^n x_iy_i$ (resp. $(\mathbf{x},\mathbf{y})_H=\sum_{i=1}^n x_iy_i^q$). 
Then the Euclidean (resp. Hermitian) dual code of $\C$ is defined as 
$\C^{\bot_E}=\{\mathbf{x}\in \F_q^n:\ (\mathbf{x},\mathbf{y})_E=0\ {\rm for\ all}\ \mathbf{y}\in \C\}$ 
(resp. $\C^{\bot_H}=\{\mathbf{x}\in \F_q^n:\ (\mathbf{x},\mathbf{y})_H=0\ {\rm for\ all}\ \mathbf{y}\in \C\}$). 
The code $\C$ is Euclidean (resp. Hermitian) \emph{formally self-dual (FSD)} if it has 
the same weight distribution as its dual code with respect to the Euclidean (resp. Hermitian) 
inner product. For more details on weight distribution, readers are referred to \cite{MPS2022,MS1977,MVL2022,ZSO2022}. %and the references therein. 
Furthermore, let $\mathbf{x}\in \F_2^n$, $\mathbf{x}$ is called \emph{even-like} if $\sum_{i=1}^nx_i=0$ and 
\emph{odd-like} otherwise. For a binary code $\C$, $\C$ is said to be \emph{even-like} if 
all of its codewords are even-like and \emph{odd-like} otherwise.

\subsection{Propagation rules}
 
One classical way of constructing  linear codes with good parameters  is to use propagation rules
which allows one to derive new linear codes from one or more known linear codes. These techniques include
 extending, lengthening, shortening, puncturing, etc., see for example \cite{LX2004}. 
When considering the combination of two or more known linear codes, one can use the 
concatenated construction, direct sum construction, $(\mathbf{u},\mathbf{u+v})$-construction, 
$(\mathbf{u+v+w},\mathbf{2u+v},\mathbf{u})$-construction, and others \cite{KP1992,LX2004,MS1977}. 
In 2001, Blackmore and Norton \cite{BN2001} introduced the so-called matrix-product codes, 
which generalized these combinatorial constructions when the lengths of given linear codes are the same. 
%the direct sum construction, $(\mathbf{u},\mathbf{u+v})$-construction, $(\mathbf{u+v+w},\mathbf{2u+v},\mathbf{u})$-construction, $(\mathbf{u+v},\mathbf{u-v})$-construction, $(\mathbf{a+x},\mathbf{b+x},\mathbf{a+b+x})$-construction and so forth.  
Note also that these propagation rules are widely 
used in quantum stabilizer codes, entanglement-assisted quantum error-correcting codes, 
subsystem codes, locally repairable codes, constant-weight codes, linear complementary dual 
(LCD) codes, etc., see for example \cite{AK2008,CXY2010,GR2015,I2022,LEGL2022,LEL2022,LEL2022LRC,LLY2022,LMWX2022} 
and the references therein. 
It is worth noting that recently the $(\mathbf{u},\mathbf{u+v})$-construction has also 
been used to design a novel Niederreiter-Like cryptosystem \cite{MCAK2021}. 
%\st{All these works illustrate the great advantage of propagation rules.}

\subsection{Hulls and related codes} 

The Euclidean (resp. Hermitian) hull of a $q$-ary linear code $\C$ is defined as 
$\Hull_E(\C)=\C\cap \C^{\bot_E}$ (resp. $\Hull_H(\C)=\C\cap \C^{\bot_H}$). 
In 1990, Assmus et al. \cite{AK1990} first introduced the hulls of linear codes to 
classify finite projective planes. %\st{Nowadays, numerous celebrated studies have indicated that}
The knowledge of the hull of a linear code  has many important applications in coding theory, such as 
determining the complexity of algorithms for computing the automorphism group of 
a linear code \cite{L1982} and for checking the permutation equivalence of two linear
codes \cite{L1991,S2000}, as well as for computing the numbers of shared pairs 
(i.e., the parameter $c$) in an entanglement-assisted quantum error-correcting code \cite{GJG2018,LLY2020}. 
Note that some special cases of  hulls   are of much interest, 
namely, $\Hull_E(\C)=\{0\}$ (resp. $\Hull_H(\C)=\{0\}$), $\Hull_E(\C)=\C$ (resp. $\Hull_H(\C)=\C$) 
and $\Hull_E(\C)=\C=\C^{\bot_E}$ (resp. $\Hull_H(\C)=\C=\C^{\bot_H}$), which correspond respectively 
to Euclidean (resp. Hermitian) LCD codes, Euclidean (resp. Hermitian) self-orthogonal codes and 
Euclidean (resp. Hermitian) self-dual codes. Further, if $\C$ is an $[n,k]_q$ Euclidean 
(resp. Hermitian) self-orthogonal code satisfying $n=2k$, then $\C$ is an Euclidean (resp. Hermitian) 
self-dual code, which provides an alternative definition of Euclidean (resp. Hermitian) self-dual codes.

In fact, LCD codes were introduced by Massey \cite{M1992} due to their asymptotically good property and 
they provide an optimum linear coding solution for two-user binary adder channel. LCD codes were  
further developed by Carlet et al. \cite{CG2016} to fight against side channel attacks (SCA) 
and against fault injection attacks (FIA). In 2018, Carlet et al. \cite{CMTQ2018} proved that 
any linear code over $\F_q$ is equivalent to some Euclidean LCD code for $q\geq 4$ and any linear 
code over $\F_{q^2}$ is equivalent to some Hermitian LCD code for $q\geq 3$. From then on, the 
focus has been on the binary, ternary Euclidean LCD codes and the quaternary Hermitian LCD codes 
with good parameters, {see for example}  \cite{AH2020,AHS2020,AHS2021.1,GKLRW2018,LLY2022,LZYC2020,ZLLL2020} and 
the references therein. %\st{What needs to be emphasized is} {\color{red}Note}  
Note that some propagation rules included puncturing 
and shortening methods based on one given linear code were applied to obtain good LCD codes, which 
improved the results in previous literature. 
%More works related to LCD codes will be presented later as needed. 

%\textcolor{red}{Self-dual and self-orthogonal codes need to be introduced here!!!} 
In addition, self-dual codes and self-orthogonal codes had also been extensively studied for their applications 
to combinatorial $t$-designs, Euclidean or Hermitian lattices, modular forms, DNA codes and quantum 
information theory \cite{BDHO1999,CS1998,KO2022,NRS2006}. Moreover, self-dual codes and self-orthogonal 
codes are asymptotically good in general \cite{D2009,MST1972}. In particular, some classification 
results on binary self-dual codes and binary self-orthogonal codes were given in 
\cite{H2007,KO2022,SLK2022} and the references therein.

\subsection{Our motivation and contribution}

Based on their practical applications pointed above, it is natural to ask whether it is possible to yield good 
linear codes with easily determined Euclidean and Hermitian hull dimensions by using propagation rules 
based on multiple linear codes. 
For this aim, we will  focus on combinations of two given linear codes in this paper. We study the direct sum construction 
and the $(\mathbf{u},\mathbf{u+v})$-construction, and determine their Euclidean and Hermitian hull dimensions 
under certain conditions. Furthermore, we present some criteria to judge whether resulting codes derived 
from these two propagation rules are self-orthogonal, self-dual and LCD codes under both Euclidean and 
Hermitian cases. As a consequence, we apply our results to construct linear codes with prescribed hull 
dimensions including Euclidean and Hermitian self-orthogonal codes, which are (almost) optimal 
for many of them according to Database \cite{Grass}. In addition, many new binary, ternary Euclidean 
FSD LCD codes, quaternary Hermitian FSD LCD codes, even-like codes, and odd-like codes are obtained. 
Some explicit examples that can be used to explain the possible effects of our results in improving 
the lower bounds on the minimum distance of LCD codes are also found. 

\subsection{Organisation of the paper}

After this introduction in Section~\ref{sec1}, we review some useful basic notions and results in Section \ref{sec2}. 
In Section \ref{sec.Propagations rules}, we study the direct sum construction and 
the $(\mathbf{u},\mathbf{u+v})$-construction as well as their Euclidean and Hermitian hull dimensions. 
Section \ref{sec.Applications} shows the construction of   linear codes with prescribed hull dimensions. 
We find new binary, ternary Euclidean FSD LCD codes, quaternary Hermitian FSD LCD codes, even-like codes, odd-like codes, and good 
 Euclidean (resp. Hermitian) 
self-orthogonal and  Euclidean (resp. Hermitian) self-dual codes based on the results in Section \ref{sec.Propagations rules}. 
Several specific examples are also given to show that our result is potentially effective in improving the lower 
bounds on the minimum distance of LCD codes.
Finally, Section \ref{sec5} shows {some} concluding remarks.

\section{Preliminaries}\label{sec2}

Throughout this paper, the corresponding finite field $\F_q$ satisfies $q=p^h$ as before when we consider 
the Euclidean inner product, and %\st{the corresponding finite field $\F_q$ satisfies $q=p^h$ with even} 
$h$ will be even when we consider the Hermitian inner product. 

\subsection{Linear intersection pairs}

%\st{In this subsection, we state the following results on linear $\ell$-intersection pairs for a later use.} 

\begin{definition}[\cite{GGJT2020}]\label{def.linear l-intersection} 
	For an integer $\ell\geq 0$, the pair of linear codes $\C$ and $\mathcal{D}$ of length $n$ over $\F_q$ is called 
	a \emph{linear $\ell$-intersection pair} if $\dim(\C\cap \mathcal{D})=\ell$.
\end{definition}

\begin{lemma}[Theorem\ 2.1 in \cite{GGJT2020}]\label{lem.dim.linear l-intersection} 
	Let $\C_i$ be an $[n,k_i,d_i]_q$ linear code with parity check matrix $H_i$ and generator matrix $G_i$ for $i=1,\ 2$. 
	If $\C_1$ and $\C_2$ is a linear $\ell$-intersection pair, 
	then $\rank(H_1G_2^T)$ and $\rank(G_1H_2^T)$ are independent of the particular election of $H_i$ and $G_i$ and 
	$$\rank(G_1H_2^T)=\rank(H_2G_1^T)=k_1-\ell,\quad 
	\rank(G_2H_1^T)=\rank(H_1G_2^T)=k_2-\ell.$$
\end{lemma}

Suppose now that $h$ is even, for any subset $S$ of $\F_{q}^n$, %\st{each element $\mathbf{v}\in S$ is an vector of $\F_{q}^n$, denoted by $\mathbf{v}=(v_1,v_2,\dots,v_n)$.} 
we define $S^{\sqrt[]{q}}=\{\mathbf{v}^{\sqrt[]{q}}:\ \mathbf{v}\in S\}$. 
It is clear that $\C^{\bot_H}=(\C^{\sqrt[]{q}})^{\bot_E}$ and $(\C^{\bot_H})^{\bot_H}=\C$. 
We denote by  $A^\dagger=[a_{ji}^{\sqrt[]{q}}]$ the conjugate transpose of a matrix 
$A=[a_{ij}]$ over $\F_q$.  The following result holds. 

\begin{corollary}\label{coro.dim.linear l-intersection}
	Let $\C_i$ be an $[n,k_i,d_i]_q$   linear code and $G_i$ be a generator matrix of $\C_i$ for $i=1,\ 2$. 
	Then the following statements hold. 
	\begin{enumerate}
		\item For the Euclidean inner product, we have $\rank(G_1G_2^T)$ and $\rank(G_2G_1^T)$ 
		are independent of $G_1$ and $G_2$ so that 
		$$\rank(G_1G_2^T)=\rank(G_2G_1^T)=k_1-\dim(\C_1\cap \C_2^{\bot_E})\hbox{ and}$$      
		$$\rank(G_2G_1^T)=\rank(G_1G_2^T)=k_2-\dim(\C_1^{\bot_E}\cap \C_2).$$
		
		\item For the Hermitian inner product, we have $\rank(G_1G_2^\dagger)$ and $\rank(G_2G_1^\dagger)$ 
		are independent of $G_1$ and $G_2$ so that 
		$$\rank(G_1G_2^\dagger)=\rank(G_2G_1^\dagger)=k_1-\dim(\C_1\cap \C_2^{\bot_H})\hbox{ and}$$ 
		$$\rank(G_2G_1^\dagger)=\rank(G_1G_2^\dagger)=k_2-\dim(\C_1^{\bot_H}\cap \C_2).$$
	\end{enumerate}
\end{corollary}
\begin{proof}
	We prove only the Hermitian case in {$2$)}, and a similar proof can be done for the Euclidean case in {$1$)}. 
	For $i=1,\ 2$, since $G_i$ is a generator matrix of $\C_i$, we have $G_i^{\sqrt{q}}$ is a generator matrix of 
	$\C_i^{\sqrt{q}}$. It follows from $\C_i^{\bot_H}=(\C_i^{\sqrt{q}})^{\bot_E}$ that $G_i^{\sqrt{q}}$ is a parity 
	check matrix of $\C_i^{\bot_H}$. Combining Definition~\ref{def.linear l-intersection}, we can derive { from Lemma~\ref{lem.dim.linear l-intersection}} both equalities.
	%that 
%	$$\rank(G_1G_2^\dagger)=\rank(G_2G_1^\dagger)=k_1-\dim(\C_1\cap \C_2^{\bot_H})$$ and 
%	$$\rank(G_2G_1^\dagger)=\rank(G_1G_2^\dagger)=k_2-\dim(\C_1^{\bot_H}\cap \C_2),$$ which completes the proof. 
\end{proof}

If we take $\C_1=\C_2=\C$, since $\rank(G_1G_2^\dagger)$ is independent of $G_1$ and $G_2$, then the following results 
in \cite[Propositions 3.1 and 3.2]{GJG2018} can also be obtained from 
Corollary \ref{coro.dim.linear l-intersection} as special cases. 

\begin{lemma}[Propositions 3.1 and 3.2 in  \cite{GJG2018}]\label{lem.Hermitian hulls} 
	Let $\C$ be an $[n,k,d]$ linear code. Assume that $G$ is a generator matrix of $\C$. Then the following statements hold. 
	\begin{enumerate}
		\item  For the Euclidean case, we have $\rank(GG^T)=k-\dim(\Hull_E(\C))=k-\dim(\Hull_E(\C^{\bot_E}))$. 
		
		\item  For the Hermitian case, we have $\rank(GG^\dagger)=k-\dim(\Hull_H(\C))=k-\dim(\Hull_H(\C^{\bot_H}))$. 
	\end{enumerate}
\end{lemma}

\iffalse
Thanks to the work of Chen \cite{C2022} and Luo et al. \cite{LEGL2022}, once we identify the dimension of the Hermitian 
hull of $\C$, it is convenient to obtain linear codes with Hermitian hulls of smaller dimensions from the following lemma. 

\begin{lemma}\label{lem.arbitrary Hermitian hull from known dimension}{\rm (\cite{C2022,LEGL2022})}
	Let $q\geq 3$ be a prime power and $\C$ be an $[n,k,d]_{q^2}$ linear code with $\dim(\Hull_H(\C))=\ell$. 
	Then there exists an $[n,k,d]_{q^2}$ linear code $\C'$ with $\dim(\Hull_H(\C))=\ell'$ for each integer 
	$0\leq \ell'\leq \ell$.   
\end{lemma}
\fi

%\st{At the end of this subsection, we recap some important propositions, which are going to contribute to our target.}  
 We will now recall some results that were initially enunciated for the $l$-Galois inner product in \cite{LP2020}. 
Since the $l$-Galois inner product is coincident with the Euclidean inner product when $l=0$ and the Hermitian 
inner product when $l=\frac{h}{2}$ with even $h$ \cite{FZ2017} we have the following.

\begin{proposition}[Corollaries 3.5 and 3.6 in \cite{LP2020}]\label{prop1}  
	Let $\C_i$ be an $[n_i,k_i,d_i]_q$ linear code and $G_i$ be a generator matrix of $\C_i$ for $i=1,\ 2$. 
	Then the following properties hold. 
	\begin{enumerate}
		%\item  $(\C_1\cap \C_2)^{\bot_E}=\C_1^{\bot_E}+\C_2^{\bot_E}$ (resp. $(\C_1\cap \C_2)^{\bot_H}=\C_1^{\bot_H}+\C_2^{\bot_H}$).
		
		\item  $\C_1$ is an Euclidean (resp. Hermitian) self-orthogonal code if and only if $G_1G_1^T$ (resp. $G_1G_1^\dagger$) is the all zeros matrix. 
		
		\item $\C_1$ is an Euclidean (resp. Hermitian) LCD code if and only if $G_1G_1^T$ (resp. $G_1G_1^\dagger$) is a non-singular matrix. 
		
		%\item [\rm (3)] If $\C_1$ is a Hermitian (resp. Eucelidean) self-orthogonal code, then $2k_1\leq n_1$.  
	\end{enumerate}
\end{proposition}

\subsection{Some particular matrices}

We say that a square matrix $M$ is a  \emph{monomial matrix} if there is exactly one nonzero entry in each row and each column, 
and all other positions are $0$. In particular, 
$M$ becomes a \emph{permutation matrix} $P$ if all of its nonzero entries are $1$. 
For two linear codes $\C_1$ and $\C_2$ with respective generator matrices $G_1$ and $G_2$, 
they are called \emph{permulation equivalent} if there exists a permutation matrix $P$ 
such that $G_2=G_1P$. It is well known that permutation equivalent codes have the same length, 
dimension and minimum distance. 
The following lemma  shows  that the dimension of the hulls of two permutation equivalent 
codes are the same. 

\begin{lemma}[Lemma 6 in \cite{LEGL2022} and Proposition 4.2 in \cite{LP2020}]\label{lem.per-equvilent hull equal}
	Let $\C_1$, $\C_2$ be any two permutation equivalent linear codes. Then 
	$\dim(\Hull_E(\C_1))=\dim(\Hull_E(\C_2))$ and $\dim(\Hull_H(\C_1))=\dim(\Hull_H(\C_2)).$ 
 
\end{lemma}

An $n\times n$ \emph{Toeplitz matrix} $T_n$ is a matrix whose all diagonals parallel to the main diagonal 
have constant entries, i.e., 
\begin{equation}\label{eq.Toeplitz matrix}
T_n=\left(\begin{array}[]{cccccc}
t & a_1 & a_2 & \cdots & \cdots & a_{n-1} \\
b_1 & t & a_1 & \ddots & \ddots & \vdots \\ 
b_2 & b_1 & \ddots & \ddots & \ddots & \vdots \\
\vdots & \ddots & \ddots & \ddots & a_1 & a_2 \\
\vdots & \ddots & \ddots & b_1 & t & a_1 \\
b_{n-1} & \cdots & \cdots & b_2 & b_1 & t \\
\end{array}\right), 
\end{equation}
where $t,a_i,b_i\in \F_q$ for $1\leq i\leq n-1$. 
$O_{s\times t}$ will denote the {all entries zero matrix} of size $s\times t$, when $s=t$, 
we abbreviate $O_{s\times t}$ as $O_s$. %\st{At last, we recall the concept of right non-singular. Let} 
A $m\times n$ matrix $A$ is \emph{right non-singular} if there is an $n\times m$ matrix 
$B$ such that $AB=I_m$. %{\color{blue}(Here the unit matrix is represented by $I_m$, while in Remark 1, it is represented by $Id_2$, and they seem to need to be unified. ??? Yang)}, 
%And $B$ is called a \emph{right inverse} of $A$. 
%In this case, denote a right inverse of $A$ by $A^{-1}$. 

\subsection{Matrix-product codes} 

In this subsection, we review some notions on matrix-product codes the reader can refer to \cite{BN2001,LLY2022,LP2020} for further details. 
\begin{definition}[Definition 5.2 in \cite{LP2020}]\label{Def.martix product code}
	Let $A=[a_{ij}]$ be an $m\times n$ matrix over $\F_q$ and $\C_1,\C_2,\cdots,\C_m$ be codes of length $n$ over $\F_q$. 
	The \emph{matrix-product code} $[\C_1,\C_2,\cdots,\C_m]\cdot A$ is given by all the vectors $ 
	[\sum_{i=1}^{m}\mathbf{c}_ia_{i1},\sum_{i=1}^{m}\mathbf{c}_ia_{i2},\cdots,\sum_{i=1}^{m}\mathbf{c}_ia_{in}]$, 
	where $\mathbf{c}_i\in \C_i$ is a row vector of length $n$ for $i=1,2,\cdots,m$. 
\end{definition}

\iffalse
When $A$ is right non-singular, we have the following well-known result. 
\begin{proposition}\label{prop.MPdim}{\rm (\cite{BN2001}, Proposition 2.9)} 
	Let $\C_1,\C_2,\cdots,\C_m$ be $m$ codes of length $n$ over $\F_q$. If an $m\times n$ matrix $A$ is right non-singular 
	then $|[\C_1,\C_2,\cdots,\C_m]\cdot A|=|\C_1| |\C_2| \cdots |\C_m|$. Furthermore, if $\C_1,\C_2,\cdots,\C_m$ be linear 
	codes, then $\dim([\C_1,\C_2,\cdots,\C_m]\cdot A)=\sum_{i=1}^m\dim(\C_i)$. 
\end{proposition}
\fi

Liu et al. determined the $l$-Galois hulls of matrix-product codes in \cite[Theorems 5.6 and 5.11]{LP2020} 
as well as their dimensions in \cite[Corollary 5.12]{LP2020} when the matrix $A \sigma^l(A^T)={\rm diag}(\lambda_1,\lambda_2)$ 
or the matrix $A\sigma^l(A^T)$ is upper triangular or lower triangular, where $\sigma$ is the Frobenius automorphism of $\F_q$. 
%Recall that if $l=0$, the $l$-Galois hull is just the Euclidean hull; if $l=\frac{h}{2}$ with even $h$, the $l$-Galois 
%hull is the Hermitian hull. %\st{Then the following result can be derived shortly.} 

\begin{lemma}[Corollary 5.12 in \cite{LP2020}]\label{lem.MP.HULL.DIM}
	Let $\C_1,\C_2,\cdots,\C_m$ be linear codes of length $n$ over $\F_q$ and $\C=[\C_1,\C_2,\cdots,\C_m]\cdot A$ be 
	the matrix-product code, where $A$ is an $m\times n$ matrix. Then the following statements hold. 
	\begin{enumerate}
		\item  For the Euclidean case, if $A$ is right non-singular and $AA^T$ is block upper triangular 
		or block lower triangular, then  
		\begin{equation}\label{eq.MP.HULL.DIM1}
		0\leq \dim(\Hull_E(\C))\leq \sum_{i=1}^m\dim(\Hull_E(\C_i)),
		\end{equation}
		and if $AA^T=\diag(\lambda_1,\lambda_2,\cdots,\lambda_m)$, where $\lambda_i\neq 0$ for all $1\leq i\leq m$, then 
		\begin{equation}\label{eq.MP.HULL.DIM2}
		\dim(\Hull_E(\C))=\sum_{i=1}^m\dim(\Hull_E(\C_i)). 
		\end{equation}
		%$$\dim(\Hull_E(\C))=\dim(\Hull_E(\C_1))+\dim(\Hull_E(\C_2))+\cdots+\dim(\Hull_E(\C_m))$$
		\item  For the Hermitian case, if $A$ is right non-singular and $AA^\dagger$ is block upper triangular 
		or block lower triangular, then  
		\begin{equation}\label{eq.MP.HULL.DIM3}
		0\leq \dim(\Hull_H(\C))\leq \sum_{i=1}^m\dim(\Hull_H(\C_i)), 
		\end{equation}
		and if $AA^\dagger=\diag(\lambda_1,\lambda_2,\cdots,\lambda_m)$, where $\lambda_i\neq 0$ for all $1\leq i\leq m$, then 
		\begin{equation}\label{eq.MP.HULL.DIM4}
		\dim(\Hull_H(\C))=\sum_{i=1}^m\dim(\Hull_H(\C_i)). 
		\end{equation}
	\end{enumerate}
\end{lemma}
%\begin{proof}
	%Let $l=0$ in [\cite{LP2020}, Corollary 5.8].
	%From [\cite{LP2020}, Corollary 5.8], we have $\Hull_E(\C)=[\Hull_E(\C_1),\Hull_E(\C_2),\cdots,\Hull_E(\C_m)]\cdot A$. 
%	Take $l=0$ and $\frac{h}{2}$ with even $h$ in \cite[Corollary 5.12]{LP2020}, respectively, 
%	then the desired results {\color{red}1 and 2 hold}. 
%\end{proof}

\subsection{Generalized Reed-Solomon  codes}
Now, let us introduce some basic notations and results on generalized Reed-Solomon (GRS) codes, which will 
be needed later. Let $a_1,a_2,\dots,a_n$ be $n$ distinct elements of $\F_q$ and $v_1,v_2,\dots,v_n$ be $n$ 
nonzero elements of $\F_q$, where $n$ is an integer satisfying $1<n\leq q$. 
Denote $\mathbf{a}=(a_1,a_2,\dots,a_n)$ and $\mathbf{v}=(v_1,v_2,\dots,v_n)$. 
For $1\leq k\leq n$, we define the $k$-dimensional GRS code of length $n$ associated with $\mathbf{a}$ and 
$\mathbf{v}$ as 
\begin{equation}\label{eq. GRS.df}
\GRS_k(\mathbf{a},\mathbf{v})=\{(v_1f(a_1),v_2f(a_2),\dots,v_nf(a_n)):\ f(x)\in \F_q[x],\ \deg(f(x))\leq k-1\}.
\end{equation}
It is obvious that $\GRS_k(\mathbf{a},\mathbf{v})$ has a generator matrix 
\begin{equation}\label{eq.GRS.generator matrix}
G=\left(\begin{array}{cccc}
v_1 & v_2 & \cdots & v_n \\
v_1a_1 & v_2a_2 & \cdots & v_na_n \\
v_1a^2_1 & v_2a^2_2 & \cdots & v_na^2_n \\
\vdots & \vdots & \ddots & \vdots \\
v_1a^{k-1}_1 & v_2a^{k-1}_2 & \cdots & v_na^{k-1}_n \\
\end{array}\right).
\end{equation}
Note also that, since $\C^{\bot_H}=(\C^q)^{\bot_E}$, we have 
\begin{equation}\label{eq.GRS.Hermitian dual.df}
%\begin{array}{cc}
\GRS_k(\mathbf{a},\mathbf{v})^{\bot_H}  = (\GRS_k(\mathbf{a},\mathbf{v})^q)^{\bot_E} % & \\
%& = \{(v_1f(a_1),v_2f(a_2),\dots,v_nf(a_n))^q:\ f(x)\in \F_{q^2}[x],\ \deg(f(x))\leq k-1\}^{\bot_E} \\
%& = \{(v^q_1f(a^q_1),v^q_2f(a^q_2),\dots,v^q_nf(a^q_n)):\ f(x)\in \F_{q^2}[x],\ \deg(f(x))\leq k-1\}^{\bot_E} \\
%& 
= \GRS_k(\mathbf{a^q},\mathbf{v^q})^{\bot_E}.
%\end{array}
\end{equation}
Note that the Euclidean and Hermitian dual codes of a GRS code are also GRS codes. 

\begin{lemma}[Theorem 9.1.6 in \cite{LX2004}]\label{lem.GRS.Euclidean}
	The Euclidean dual of the GRS code $\GRS_k(\mathbf{a},\mathbf{v})$ over $\F_q$ of length $n$ is 
	$\GRS_{n-k}(\mathbf{a},\mathbf{v'})$ for some $\mathbf{v'}\in (\F_q^*)^n$. 
\end{lemma}

\begin{lemma}\label{lem.GRS.Hermitian}%{\rm \cite[Theorem 9.1.6]{LX2004}}
	The Hermitian dual of the GRS code $\GRS_k(\mathbf{a},\mathbf{v})$ over $\F_{q^2}$ of length $n$ is 
	$\GRS_{n-k}(\mathbf{a^q},\mathbf{(v^q)'})$ for some $\mathbf{(v^q)'}\in (\F_{q^2}^*)^n$. 
\end{lemma}
\begin{proof}
	Combining  Equation~(\ref{eq.GRS.Hermitian dual.df}) and Lemma \ref{lem.GRS.Euclidean}, one can easily 
	verify the desired result. 
\end{proof}
Finally, 
from the generator matrix showed in Equation (\ref{eq.GRS.generator matrix}), we can derive the following 
proposition directly. 

\begin{proposition}\label{prop.GRS nested}
	For two GRS codes $\GRS_{k_1}(\mathbf{a},\mathbf{v})$ and $\GRS_{k_2}(\mathbf{a},\mathbf{v})$, 
	we have $\GRS_{k_1}(\mathbf{a},\mathbf{v})\subseteq \GRS_{k_2}(\mathbf{a},\mathbf{v})$ 
	if and only if $k_1\leq k_2$. 
\end{proposition}

\subsection{The Griesmer's bound}

For an $[n,k,d]_q$ linear code, there are many trade-offs among its parameters 
$n$, $k$ and $d$.  %\st{the commonly used bounds is} 
{One of them is the so-called Griesmer's bound. % and one can find a detailed proof in \cite[Theorem 5.7.4]{LX2004}.} %\st{We rephase it in the following lemma, which will be needed later. }
\begin{lemma}[Theorem 5.7.4 in \cite{LX2004}]\label{lem.Griesmer bound} 
	Let $\C$ be a $q$-ary code of parameters $[n,k,d]$, where $k\geq 1$. Then 
	\begin{equation}\label{eq. Griesmer bound}
	n\geq \sum_{i=0}^{k-1} \lceil \frac{d}{q^i} \rceil. 
	\end{equation}
\end{lemma}
%\st{It is well known that} 
If a linear code $\C$ meets {its Griesmer's bound}, we also call it \emph{optimal}.

\section{Two classical propagation rules}\label{sec.Propagations rules}

In this section, we recall two classical propagation rules based on two given linear codes and 
determine the dimension of the hull of the linear codes resulting from them. %\st{What needs to be emphasized is that,}  
Throughout this section, we always prove all the results in the Hermitian 
case when we take into account both the Euclidean inner product and the Hermitian inner product. 
Indeed one can also prove the same results for the Euclidean case, and thus, we omit those similar proofs. 

\subsection{The direct sum construction}	
\begin{theorem}\label{th.propagation1.direct sum}
	Let $\C_i$ be an $[n_i,k_i,d_i]_{q}$ linear code and $G_i$ be a generator matrix of $\C_i$ for $i=1,\ 2$. 
	The direct sum of $\C_1$ and $\C_2$ is defined by 
	$
	\C_1\oplus \C_2=\{(\mathbf{c}_1,\mathbf{c}_2):\ \mathbf{c}_1\in \C_1,\ \mathbf{c}_2\in \C_2\}.
$
	Then the following statements hold. 
	\begin{enumerate}
		\item  $\C_1\oplus \C_2$ is an $[n_1+n_2,k_1+k_2,\min\{d_1,d_2\}]_{q}$ linear code with a generator matrix 
		$$G=\left( \begin{array}{cc}
		G_1 & O_{k_1\times n_2} \\
		O_{k_2\times n_1} & G_2 \\
		\end{array} \right).$$
		
		\item  For the Euclidean case, suppose that $\dim(\Hull_E(\C_1))=l_1$ and $\dim(\Hull_E(\C_2))=l_2$. Then 
		\begin{equation}\label{eq.direct sum EHull}
		\dim(\Hull_E(\C_1\oplus \C_2))=l_1+l_2.
		\end{equation} 
		
		\item  For the Hermitian case, suppose that $\dim(\Hull_H(\C_1))=l_1$ and $\dim(\Hull_H(\C_2))=l_2$. Then 
		\begin{equation}\label{eq.direct sum HHull}
		\dim(\Hull_H(\C_1\oplus \C_2))=l_1+l_2.
		\end{equation}
	\end{enumerate}
\end{theorem}
\begin{proof}
	{ 1)} This is a well-known result. 

	{ 2)} Similar to the approach in {  3)} below, {  2)} holds. 
		 
	{ 3)} Since $\dim(\Hull_H(\C_1))=l_1$ and $\dim(\Hull_H(\C_2))=l_2$, by Lemma \ref{lem.Hermitian hulls}, we have 
	$\rank(G_1G_1^\dagger)=k_1-l_1$ and $\rank(G_2G_2^\dagger)=k_2-l_2$. Moreover, it follows from  
	\begin{equation*}\begin{split}
	GG^\dagger & = \left( \begin{array}{cc}
	G_1 & O_{k_1\times n_2} \\
	O_{k_2\times n_1} & G_2 \\
	\end{array} \right) 
	\left( \begin{array}{cc}
	G_1^\dagger & O_{n_1\times k_2} \\
	O_{n_2\times k_1} & G_2^\dagger \\
	\end{array} \right)\\
	& = \left( \begin{array}{cc}
	G_1G_1^\dagger & O_{k_1\times k_2} \\
	O_{k_2\times k_1} & G_2G_2^\dagger \\
	\end{array} \right)
	\end{split}\end{equation*}
	that $\rank(GG^\dagger)=\rank(G_1G_1^\dagger)+\rank(G_2G_2^\dagger)=k_1+k_2-l_1-l_2$, i.e., 
	$$\dim(\Hull_H(\C_1\oplus \C_2))=k_1+k_2-\rank(GG^\dagger)=l_1+l_2.$$
%	which completes the proof. 
	%$$G_1G_1^\dagger=O_{k_1\times k_1}\ {\rm and}\ G_2G_2^\dagger=O_{k_2\times k_2}$$
\end{proof}
\begin{remark}\label{rem1.direct sum}
	%\st{It is known easily that} 
	When the linear codes $\C_1$ and $\C_2$ have the same length, 
	the direct sum construction associated to $\C_1$ and $\C_2$ is exactly the matrix-product code 
	$[\C_1,\C_2]\cdot A$, where $A=\left(\begin{array}[]{cc}
	1 & 0 \\
	0 & 1 \\
	\end{array}\right)$. In this case, $AA^T=AA^\dagger=I_2$ and the dimension of the hull of $\C_1\oplus \C_2$ follows from 
	from Lemma \ref{lem.MP.HULL.DIM}. 
	However, in Theorem \ref{th.propagation1.direct sum}, we can also consider the case where the 
	lengths of $\C_1$ and $\C_2$ are different, thus  
	it is a generalization of Lemma \ref{lem.MP.HULL.DIM}. 
\end{remark}
The following results can be deduced from Theorem \ref{th.propagation1.direct sum}.

\begin{corollary}\label{coro1.th.propagation1}
	With the same notation as in Theorem \ref{th.propagation1.direct sum}, the following statements hold. 
	\begin{enumerate}
		\item  $\C_1\oplus \C_2$ is an Euclidean (resp. Hermitian) self-orthogonal code if and only if 
		both $\C_1$ and $\C_2$ are Euclidean (resp. Hermitian) self-orthogonal codes. 
		\item  $\C_1\oplus \C_2$ is an Euclidean (resp. Hermitian) LCD code if and only if 
		both $\C_1$ and $\C_2$ are Euclidean (resp. Hermitian) LCD codes.  
		\item  $\C_1\oplus \C_2$ is an Euclidean (resp. Hermitian) self-dual code if and only if 
		both $\C_1$ and $\C_2$ are Euclidean (resp. Hermitian) self-dual codes.
	\end{enumerate}
\end{corollary}
\begin{proof}
	{ 1)} By Theorem \ref{th.propagation1.direct sum}.{3)}, we can conclude that $\C_1\oplus \C_2$ 
	is a Hermitian self-orthogonal code if and only if $k_1+k_2=l_1+l_2$. Since $0\leq l_1\leq k_1$ and 
	$0\leq l_2\leq k_2$, we have $k_1+k_2=l_1+l_2$ if and only if $l_1=k_1$ and $l_2=k_2$, which 
	is equivalent to both $\C_1$ and $\C_2$ being Hermitian self-orthogonal codes. 
	Therefore, the result { 1)} holds. 
	
	{ 2)} Similar to {  1)} above, we can conclude that $\C_1\oplus \C_2$ is a Hermitian LCD code if and only if $l_1+l_2=0$, 
	if and only if $l_1=l_2=0$, which is equivalent to both $\C_1$ and $\C_2$ being Hermitian LCD codes. 
	Therefore, the result {  2)} holds. 
	
	{  3)} We note that $\C_1\oplus \C_2$ is a Hermitian self-dual code if and only if $\C_1\oplus \C_2$ is a Hermitian 
	self-orthogonal code and $n_1+n_2=2k_1+2k_2$. By { 1)} above, $\C_1\oplus \C_2$ is a Hermitian self-orthogonal code 
	if and only if both $\C_1$ and $\C_2$ are Hermitian self-orthogonal codes, which yield $2k_1\leq n_1$ and $2k_2\leq n_2$, 
	respectively. It follows that $n_1+n_2=2k_1+2k_2$ if and if only $n_1=2k_1$ and $n_2=2k_2$. 
	It follows that $\C_1\oplus \C_2$ is a Hermitian self-dual code if and only if both $\C_1$ and $\C_2$ are Hermitian self-dual 
	codes. Therefore, the result {  3)} holds. 
\end{proof}
\begin{remark}\label{rem2.propagation1}
	In \cite[Theorem 3.1]{LS2022}, the authors proved any linear code $\C$ over $\F_{q}$ can be seen as the 
	direct sum of an Euclidean (resp. Hermitian) self-orthogonal code and an Euclidean (resp. Hermitian) LCD code. 
	In Corollary \ref{coro1.th.propagation1}, we give some criteria to identify whether the direct sum of two linear 
	codes is an Euclidean (resp. Hermitian) self-orthogonal code, an Euclidean (resp. Hermitian) LCD code or 
	an Euclidean (resp. Hermitian) self-dual code. 
\end{remark}

\iffalse
\begin{corollary}\label{coro2.th.propagation1}
	Let notations be the same as Theorem \ref{th.propagation1.direct sum}. Suppose that $\dim(\Hull_H(\C_1))=l_1$ and $\dim(\Hull_H(\C_2))\\=l_2$. 
	Then there exists an $[n_1+n_2,k_1+k_2,\min\{d_1,d_2\}]_{q^2}$ linear code $\C'$ with $\dim(\Hull_H(\C'))=\ell'$ for each integer $0\leq \ell'\leq l_1+l_2$.  
\end{corollary}
\begin{proof}
	Combining Lemma \ref{lem.arbitrary Hermitian hull from known dimension} and Theorem \ref{th.propagation1.direct sum}(3), the desired result holds. 
\end{proof}
\fi

\subsection{The $(\mathbf{u},\mathbf{u+v})$-construction}

\begin{theorem}\label{th.propagation1.(u,u+v)-construction}
	{ Let $\C_i$ be an $[n,k_i,d_i]_q$ linear code and $G_i$ be a generator matrix of $\C_i$ for $i=1,\ 2$, 
		and $\C$ the linear code defined by }
	$
	\C=\{(\mathbf{u},\mathbf{u+v}):\ \mathbf{u}\in \C_1,\ \mathbf{v}\in \C_2\}.
	$
	Then the following statements hold. 
	\begin{enumerate}
		\item  $\C$ is a $[2n,k_1+k_2,\min\{2d_1,d_2\}]_q$ linear code with a generator matrix 
		$$G=\left( \begin{array}{cc}
		G_1 & G_1 \\
		O_{k_2\times n} & G_2 \\
		\end{array} \right).$$
		
		\item  For the Euclidean case, suppose that $\dim(\Hull_E(\C_1))=l_1$ and $\dim(\Hull_E(\C_2))=l_2$. 
		If $q=2^m$ or $\C_1$ is an Euclidean self-orthogonal code, then 
		\begin{equation}\label{eq.uv EHull.1}
		2\dim(\C_1\cap \C_2^{\bot_E})+l_2-k_1\leq \dim(\Hull_E(\C))\leq 2\dim(\C_1\cap \C_2^{\bot_E})+k_2-k_1.
		\end{equation}
		Moreover, if $\C_2$ is an Euclidean self-orthogonal code, then 
		\begin{equation}\label{eq.uv EHull.2}
		\dim(\Hull_E(\C))=2\dim(\C_1\cap \C_2^{\bot_E})+k_2-k_1.
		\end{equation}

		\item  For the Hermitian case, suppose that $\dim(\Hull_H(\C_1))=l_1$ and $\dim(\Hull_H(\C_2))=l_2$. 
		If $q=2^m$ or $\C_1$ is a Hermitian self-orthogonal code, then 
		\begin{equation}\label{eq.uv HHull.1}
		2\dim(\C_1\cap \C_2^{\bot_H})+l_2-k_1\leq \dim(\Hull_H(\C))\leq 2\dim(\C_1\cap \C_2^{\bot_H})+k_2-k_1.
		\end{equation}
		
		Moreover, if $\C_2$ is a Hermitian self-orthogonal code, then 
		\begin{equation}\label{eq.uv HHull.2}
		\dim(\Hull_H(\C))=2\dim(\C_1\cap \C_2^{\bot_H})+k_2-k_1.
		\end{equation}
	\end{enumerate}
\end{theorem}
\begin{proof}
	{ 1)} This is a well-known result. %and one can find a proof in \cite[Theorem 6.1.8]{LX2004}. 

	{ 2)} Similar to the approach in {  3)} below, {  2)} holds.   
	
	{ 3)} Since $\dim(\Hull_H(\C_1))=l_1$ and $\dim(\Hull_H(\C_2))=l_2$, by Lemma \ref{lem.Hermitian hulls}, we have 
	$\rank(G_1G_1^\dagger)=k_1-l_1$ and $\rank(G_2G_2^\dagger)=k_2-l_2$. Moreover, we have  
\begin{equation*}
	\begin{split}
	GG^\dagger & = \left( \begin{array}{cc}
	G_1 & G_1 \\
	O_{k_2\times n} & G_2 \\
	\end{array} \right) 
	\left( \begin{array}{cc}
	G_1^\dagger & O_{n\times k_2} \\
	G_1^\dagger & G_2^\dagger \\
	\end{array} \right)\\
	& = \left( \begin{array}{cc}
	2G_1G_1^\dagger & G_1G_2^\dagger \\
	G_2G_1^\dagger & G_2G_2^\dagger \\
	\end{array} \right).
	\end{split}
\end{equation*}
	We now discuss $\rank(GG^\dagger)$. When $q=2^m$ or $\C_1$ is a Hermitian self-orthogonal code, 
	it follows that 
	$$\rank(G_1G_2^\dagger)+\rank(G_2G_1^\dagger)\leq 
	\rank(GG^\dagger)\leq \rank(G_1G_2^\dagger)+\rank(G_2G_1^\dagger)+\rank(G_2G_2^\dagger)$$
	from 
\begin{equation*}
	\begin{split}
	GG^\dagger & = \left( \begin{array}{cc}
	O_{k_1\times k_1} & G_1G_2^\dagger \\
	G_2G_1^\dagger & G_2G_2^\dagger \\
	\end{array} \right). 
	\end{split}
\end{equation*}
	We notice that $(G_1G_2^\dagger)^\dagger=G_2G_1^\dagger$, and hence, 
	from Corollary~\ref{coro.dim.linear l-intersection}.2), we can derive that 
	$\rank(G_1G_2^\dagger)+\rank(G_2G_1^\dagger)=2\rank(G_1G_2^\dagger)=2k_1-2\dim(\C_1\cap \C_2^{\bot_H})$ and 
	$\rank(G_1G_2^\dagger)+\rank(G_2G_1^\dagger)+\rank(G_2G_2^\dagger)=2k_1+k_2-2\dim(\C_1\cap \C_2^{\bot_H})-l_2$. 
	In Summary, in this case, we can conclude that 
	$2k_1-2\dim(\C_1\cap \C_2^{\bot_H})\leq \rank(GG^\dagger)\leq 2k_1+k_2-2\dim(\C_1\cap \C_2^{\bot_H})-l_2$, i.e., 
$$
	2\dim(\C_1\cap \C_2^{\bot_H})+l_2-k_1\leq \dim(\Hull_H(\C))\leq 2\dim(\C_1\cap \C_2^{\bot_H})+k_2-k_1, 
$$
	which implies that Equation~(\ref{eq.uv HHull.1}) holds. 
	Moreover, when $\C_2$ is a Hermitian self-orthogonal code, we have $l_2=k_2$. 
	Therefore, Equation~(\ref{eq.uv HHull.2}) can be deduced from Equation ~(\ref{eq.uv HHull.1}) and 
	this completes the proof. 
\end{proof}

In next theorem, we will take $\C_1$ and $\C_2$ as two special linear codes over $\F_2${, namely, $\C_1$ will be an even-like code 
	and $\C_2$ an $[n,1,n]_2$ repetition code.} 

\begin{theorem}\label{th.propagation2.(u,u+v)-construction}
	%Let ${\rm char}(\F_q)=2$. 
	Let $\C_2$ be an $[n,1,n]_2$ repetition code generated by 
	$\mathbf{1_n}=(\underbrace{1\ 1\ \cdots\ 1}_{n})$ and other notations be the same as 
	Theorem \ref{th.propagation1.(u,u+v)-construction}. 
	If $\C_1$ is an even-like code, then $\C$ is a $[2n,k_1+1,\min\{2d_1,n\}]_2$ code with  
	\begin{equation}\label{eq.EHull.3}
	\dim(\Hull_E(\C))=\left\{\begin{array}[]{ll}
	k_1, & {\rm if\ n\ is\ odd}, \\
	k_1+1, & {\rm if\ n\ is\ even}. \\
	\end{array}
	\right.
	\end{equation}
\end{theorem}
\begin{proof}
	From Theorem \ref{th.propagation1.(u,u+v)-construction}.{1)}, under the given conditions, 
	$\C$ is an $[2n,k_1+1,\min\{2d_1,n\}]_2$ linear code with a generator matrix 
	$$G=\left( \begin{array}[]{cc}
	G_1 & G_1 \\
	O_{1\times n} & \mathbf{1_n} \\
	\end{array} \right).$$ 
	
	Moreover, we have  
	\begin{equation}\label{eq1.propogation2}
	\begin{split}
	GG^T & = \left( \begin{array}{cc}
	G_1 & G_1 \\
	O_{1\times n} & \mathbf{1_n} \\
	\end{array} \right) 
	\left( \begin{array}{cc}
	G_1^T & O_{n\times 1} \\
	G_1^T & \mathbf{1_n}^T \\
	\end{array} \right)\\
	& = \left( \begin{array}{cc}
	2G_1G_1^T & G_1\mathbf{1_n}^T \\
	\mathbf{1_n}G_1^T & \mathbf{1_n}\mathbf{1_n}^T \\
	\end{array} \right).
	\end{split}
	\end{equation}  
	Since $G_1$ is a generator matrix of $\C_1$, we can set 
	$$G_1=\left( \begin{array}{cccc}
	g_{11} & g_{12} & \cdots & g_{1n} \\
	g_{21} & g_{22} & \cdots & g_{2n} \\
	\vdots & \vdots & \cdots & \vdots \\ 
	g_{k_11} & g_{k_12} & \cdots & g_{k_1 n} \\
	\end{array}
	\right),$$
	which follows that 
\begin{equation*}
	\begin{split}
	\mathbf{1_n}G_1^T & = ( G_1\mathbf{1_n}^T)^T \\
	& = (\sum_{i=1}^n g_{1i},\sum_{i=1}^n g_{2i},\cdots,\sum_{i=1}^n g_{k_1i}) \\
	%& = ((\sum_{i=1}^n g_{1i})^q\ (\sum_{i=1}^n g_{2i})^q  \ \cdots\ (\sum_{i=1}^n g_{k_1i})^q) \\
	& = O_{1\times k_1} 
	\end{split}
\end{equation*}
	from the facts that $\C_1$ is an even-like code and $q=2$. Furthermore, we also have 
	$2G_1G_1^\dagger=O_{k_1\times k_1}$ and $\mathbf{1_n}\mathbf{1_n}^T=n$. Then, $GG^T$ 
	in Equation~(\ref{eq1.propogation2}) can be further written as 
	\begin{equation}\label{eq2.propogation2}
	\begin{split}
	GG^T & = \left( \begin{array}{cc}
	O_{k_1\times k_1} & O_{k_1\times 1} \\
	O_{1\times k_1} & n \\
	\end{array} \right),  
	\end{split}
	\end{equation}
	which implies that  
	$$\rank(GG^T)=\rank(n)=\left\{\begin{array}[]{ll}
	1, & {\rm if}\ $n$\ {\rm is\ odd}, \\
	0, & {\rm if}\ $n$\ {\rm is\ even}, \\
	\end{array}
	\right. $$
	and hence, { we can derive that 
$$
		\dim(\Hull_E(\C))=\left\{\begin{array}[]{ll}
		k_1, & {\rm if\ n\ is\ odd}, \\
		k_1+1, & {\rm if\ n\ is\ even}. \\
		\end{array}
		\right. .
$$
	}	
\end{proof}
Note that the code in the proof of the result above can also be seen as a particular instance of the propagation rule in \cite{Edgar2004}.

\iffalse
[\rm (b)] Let $p\neq 2$. Suppose that $\C_1$ is an Euclidean LCD code with , then 
\begin{align}\label{eq.EHull.4}
\dim(\Hull_E(\C))=\left\{\begin{array}[]{lc}
0, & {\rm if}\ p\nmid n, \\
1, & {\rm if}\ p\nmid n. \\
\end{array}
\right.
\end{align}

\fi

\begin{remark}\label{rem3.uv}
	Let $\C_i$ be an $[n,k_i,d_i]_q$ linear code for $i=1,\ 2$.  
	As we know, by taking $A=\left(\begin{array}[]{cc}
	1 & 1 \\
	0 & 1 \\
	\end{array}\right)$, the matrix-product code $\C=[\C_1,\C_2]\cdot A$ coincides the 
	$(\mathbf{u}, \mathbf{u}+\mathbf{v})$-construction studied in Theorem \ref{th.propagation1.(u,u+v)-construction}. 
	Notice that 
	$$AA^T=AA^\dagger=\left(\begin{array}[]{cc}
	1 & 1 \\
	0 & 1 \\
	\end{array}\right)\left(\begin{array}[]{cc}
	1 & 0 \\
	1 & 1 \\
	\end{array}\right)=\left(\begin{array}[]{cc}
	2 & 1 \\
	1 & 1 \\
	\end{array}\right)$$ is neither diagonal nor triangular, and thus, the results in 
	Theorems \ref{th.propagation1.(u,u+v)-construction} and \ref{th.propagation2.(u,u+v)-construction} 
	can not be obtained from Lemma \ref{lem.MP.HULL.DIM}. 
\end{remark}

\begin{corollary}\label{coro1.th.propagation2}
	With the same notation as in Theorem \ref{th.propagation1.(u,u+v)-construction},  
	if one of the following conditions is met: 
	\begin{enumerate}
		\item  $q=2^m$ and $\C_2$ is an Euclidean (resp. Hermitian) self-orthogonal code; 
		\item  both $\C_1$ and $\C_2$ are Euclidean (resp. Hermitian) self-orthogonal codes, 
	\end{enumerate} 
	then the following statements hold. 
	\begin{enumerate}
		\item  $\C$ is an Euclidean (resp. Hermitian) self-orthogonal code if and only if 
		$\C_1\subseteq \C_2^{\bot_E}$ (resp. $\C_1\subseteq \C_2^{\bot_H}$), 
		if and only if $\C_2\subseteq \C_1^{\bot_E}$ (resp. $\C_2\subseteq \C_1^{\bot_H}$). 
		
		\item  $\C$ is an Euclidean (resp. Hermitian) LCD code if and only if 
		$2\dim(\C_1\cap \C_2^{\bot_E})=k_1-k_2$ (resp. $2\dim(\C_1\cap \C_2^{\bot_H})=k_1-k_2$). 
		
		\item  $\C$ is an Euclidean (resp. Hermitian) self-dual code if and only if 
		$\C_1=\C_2^{\bot_E}$ (resp. $\C_1=\C_2^{\bot_H}$), if and only if $\C_2=\C_1^{\bot_E}$ 
		(resp. $\C_2=\C_1^{\bot_H}$). 
	\end{enumerate}
	\begin{proof}
		{ 1)} By Theorem \ref{th.propagation1.(u,u+v)-construction}.{3)}, we can conclude that 
		$\C$ is a Hermitian self-orthogonal code if and only if $2\dim(\C_1\cap \C_2^{\bot_H})+k_2-k_1=k_1+k_2$, i.e., 
		$\dim(\C_1\cap \C_2^{\bot_H})=k_1=\dim(\C_1)$. And since $\C_1\cap \C_2^{\bot_H}\subseteq \C_1$, it 
		follows that $\C_1\cap \C_2^{\bot_H}=\C_1$, which is equivalent to $\C_1\subseteq \C_2^{\bot_H}$ and also 
		to $\C_2\subseteq \C_1^{\bot_H}$. Therefore, the result { 1)} holds. 
		
		{ 2)}  By Theorem \ref{th.propagation1.(u,u+v)-construction}.{3)}, we can conclude that 
		$\C$ is a Hermitian LCD code if and only if $2\dim(\C_1\cap \C_2^{\bot_H})+k_2-k_1=0$, i.e., 
		$2\dim(\C_1\cap \C_2^{\bot_H})=k_1-k_2$. Therefore, the result { 2)} holds. 
		
		{ 3)} We know that $\C$ is a Hermitian self-dual code if and only if $\C$ is a Hermitian self-orthogonal 
		code and $n=k_1+k_2$. By {  1)} above, from $n=k_1+k_2$, it is easy to check that $\C_1\subseteq \C_2^{\bot_H}$ 
		and $\C_2\subseteq \C_1^{\bot_H}$ are equivalent to $\C_1=\C_2^{\bot_H}$ and $\C_2=\C_1^{\bot_H}$, respectively.  
		Therefore, the result { 3)} holds.
	\end{proof} 
\end{corollary}

\begin{remark}\label{rem4.uv}
	For the Euclidean case, in \cite[Corollary IV.6]{MJ2016}, the authors gave a 
	conclusion similar to Corollary \ref{coro1.th.propagation2} 
	when $p\equiv 1\ ({\rm mod}\ 4)$, or $q$ is square and $p\equiv 3\ ({\rm mod}\ 4)$, 
	where $p$ is the characteristic of $\F_q$. For the Hermitian case, a similar conclusion 
	was also proposed in \cite[Corollary 4.9]{JM2017}. 
	{ We notice that the minimum distances of the self-orthogonal codes obtained by these two methods have only one lower bound, i.e., their minimum distances can not be determined exactly. 
	 }
	Due to these facts, we can see that the results in Corollary \ref{coro1.th.propagation2} 
	are more precise and new. 
\end{remark}

Recall that we call an $[n,k,d]_q$ linear code $\C$ Euclidean (resp. Hermitian) 
self-orthogonal if $\dim(\Hull_E(\C))=k$ (resp. $\dim(\Hull_H(\C))=k$) and further 
call it Euclidean (resp. Hermitian) self-dual if the condition $n=2k$ is also 
satisfied at this point. Notice that the $[2n,k_1+1,\min\{2d_1,d_2\}]_2$ linear 
code with $k_1$-dimensional Euclidean hull can be deduced from Theorem 
\ref{th.propagation2.(u,u+v)-construction}.

\begin{definition}\label{def.almost so and almost sd}
	{\ } We say that an $[n,k,d]_q$ linear code $\C$ is \emph{almost\ Euclidean\ self-orthogonal}  
	if $\dim(\Hull_E(\C))=k-1$. Moreover, we say that  $\C$ is \emph{almost\ Euclidean\ self-dual} 
	if $\C$ is almost Euclidean self-orthogonal and $n=2k$. 
\end{definition}

\begin{lemma}\label{coro2.th.propagation2}
	%Let $\C_2$ be an $[n,1,n]_2$ repetition code generated by $\mathbf{1_n}=(\underbrace{1\ 1\ \cdots\ 1}_{n})$ and other 
	With the same notation as in Theorem \ref{th.propagation2.(u,u+v)-construction}, the following statements hold. 
	\begin{enumerate}
		\item  $\C$ is a binary Euclidean almost self-orthogonal code 
		if and only if $n$ is odd. 
		
		\item  $\C$ is a binary Euclidean self-orthogonal code 
		if and only if $n$ is even. 
	\end{enumerate} 
\end{lemma}
\begin{proof}
	By the definition of Euclidean (almost) self-orthogonal codes, combining the 
	conclusion depicted in Theorem \ref{th.propagation2.(u,u+v)-construction}, 
	the desired results follow. 
\end{proof}

\begin{remark}\label{rem3.MPLCD} $\quad$
	\begin{enumerate}
		\item In \cite[Theorems 3.1-3.4]{LLY2022}, the authors proposed some sufficient conditions 
		(some are also necessary)  for a matrix-product code 
		$\C=[\C_1,\C_2,\cdots,\C_s]\cdot A$ being an Euclidean LCD code. Considering $A=\left(\begin{array}[]{cc}
		1 & 1 \\
		0 & 1 \\
		\end{array}\right)$ again, it is easy to verify that conditions in \cite[Theorems 3.1-3.4]{LLY2022},  
		do not hold entirely, which indicates that as a special case of matrix-product codes, 
		the $(\mathbf{u},\mathbf{u+v})$-construction is not included in \cite[Theorem 3.1-3.4]{LLY2022}. 
		In this sense, we can say that Corollary \ref{coro1.th.propagation2} is new. 
		\item  The given linear codes were required to be Euclidean LCD codes in \cite{LLY2022}. 
		However, Corollary \ref{coro1.th.propagation2} illustrates that when given linear codes 
		are not LCD codes, the resulting code may also be a LCD code under certain conditions. 
		And since Corollary \ref{coro1.th.propagation2} holds for both the Euclidean case and the Hermitian case 
		over any finite field and some criteria for the resulting codes being self-orthogonal or self-dual codes 
		are also given in Corollary \ref{coro1.th.propagation2}, we can say that Corollary \ref{coro1.th.propagation2} 
		extends the conclusions in \cite{LLY2022}. 
		\item It should be also noticed that Corollary \ref{coro2.th.propagation2} establishes the relationships 
		between binary Euclidean (almost) self-orthogonal codes derived from the $(\mathbf{u},\mathbf{u+v})$ 
		construction and the parity of code length $n$.  As far as we know, these relationships are new.  
		
		\item For the same reason with Remark \ref{rem1.direct sum}, two given linear codes of different 
		lengths can be used in Corollary \ref{coro1.th.propagation1}. And determination standards for self-orthogonal 
		and self-dual codes are presented. Therefore, in these respects, Corollary \ref{coro1.th.propagation1} can also 
		be regarded as a generalization of \cite{LLY2022}. 
	\end{enumerate}

\end{remark}

\section{Applications}\label{sec.Applications}
In this section, we employ two propagation rules in Section \ref{sec.Propagations rules} for the constructions of 
linear codes with prescribed hull dimensions, FSD LCD codes, even-like and odd-like codes as well as good  
Euclidean and Hermitian (almost) self-orthogonal codes. %The resulting codes are new and have good parameters. 
Additionally, some concrete examples are given to illustrate that our result is potentially effective in improving the lower 
bounds on the minimum distance of LCD codes. 

\subsection{Linear codes with prescribed hull dimension}

%\st{In this subsection, we present here a practical example that can help us further understand the mechanism of the propagation rules.} 
According to Theorems \ref{th.propagation1.direct sum} and \ref{th.propagation1.(u,u+v)-construction}, 
the following example suggests that we can use some known linear codes to generate new 
linear codes with prescribed hull dimension. 

\begin{example}\label{exam.mechanism of the propagation rules}
	Let $\F_4=\{0,1,\omega,\omega^2\}$, where $\omega^2+\omega+1=0$. Let matrices $G_1$ and $G_2$ be 
	$$G_1=\left(\begin{array}[]{cccccc}
	1 & 0 & \omega & 0 & 1 & \omega \\
	0 & 1 & 0 & \omega^2 & 0 & \omega \\		
	\end{array}\right)\ {\rm and}\ 
	G_2=\left(\begin{array}[]{cccccc}
	1 & 0 & 1 & 0 & \omega & \omega \\
	0 & 1 & 0 & 1 & \omega & \omega \\		
	\end{array}\right),$$
	respectively. Computed with the Magma software package \cite{BCP1997}, we have the following results: 
	\begin{itemize}
		\item The linear code $\C$ generated by $G_1$ is a $[6,2,3]_4$ Hermitian LCD code and the linear code $\D$ 
		generated by $G_2$ is a $[6,2,4]_4$ Hermitian self-orthogonal code; 
		
		\item The Hermitian dual codes of $\C$ and $\D$ have the same parameters $[6,4,2]_4$, denoted by $\C^{\bot_H}$ 
		and $\D^{\bot_H}$, respectively;
		
		\item $\dim(\Hull_H(\C))=\dim(\Hull_H(\C^{\bot_H}))=0$ and $\dim(\Hull_H(\D))=\dim(\Hull_H(\D^{\bot_H}))=2$; 
		
		\item $\C\cap \D^{\bot_H}$ is a $[6,0,6]_4$ linear code and $\C^{\bot_H}\cap \D^{\bot_H}$ is a $[6,2,3]_4$ linear code. 
		Hence, $\dim(\C\cap \D^{\bot_H})=0$ and $\dim(\C^{\bot_H}\cap \D^{\bot_H})=2$. 
	\end{itemize}
	Using the direct sum construction and the $(\mathbf{u},\mathbf{u+v})$-construction, all possible resulting linear
	codes are listed in Table \ref{tab:mechanism of the propagation rules}. 
	
	\begin{table}[t]
		\centering
		\caption{The mechanism of the two propagation rules (Linear codes with prescribed hull dimensions)}
		\label{tab:mechanism of the propagation rules}       % Give a unique label
		\begin{tabular}{ccccc}\hline
			$\mathbf{\C_1}$ & $\mathbf{\C_2}$ & {\bf Resulting codes} & {\bf Dimension of the Hermitian hull} & {\bf Propagation rule} \\ \hline
			
			$\C$ & $\C$ & $[12,4,3]_4$ & 0(Hermitian LCD code) & Theorem \ref{th.propagation1.direct sum} \\ 
			$\C$ & $\C^{\bot_H}$ & $[12,6,2]_4$ & 0(Hermitian LCD code)& Theorem \ref{th.propagation1.direct sum} \\ 
			%$\C$ & $\D$ & $[12,4,3]_4$ & 2 & Theorem \ref{th.propagation1.direct sum} \\ 
			$\C$ & $\D^{\bot_H}$ & $[12,6,2]_4$ & 2 & Theorem \ref{th.propagation1.direct sum} \\ 
			
			$\C^{\bot_H}$ & $\C^{\bot_H}$ & $[12,8,2]_4$ & 0(Hermitian LCD code) & Theorem \ref{th.propagation1.direct sum} \\ 
			%$\C^{\bot_H}$ & $\D$ & $[12,6,2]_4$ & 2 & Theorem \ref{th.propagation1.direct sum} \\ 
			$\C^{\bot_H}$ & $\D^{\bot_H}$ & $[12,8,2]_4$ & 2 & Theorem \ref{th.propagation1.direct sum} \\ 
			
			%$\D$ & $\D$ & $[12,4,4]_4$ & 4(Hermitian self-orthogonal code) & Theorem \ref{th.propagation1.direct sum} \\ 
			$\D$ & $\D^{\bot_H}$ & $[12,6,2]_4$ & 4 & Theorem \ref{th.propagation1.direct sum} \\ 
			$\D^{\bot_H}$ & $\D^{\bot_H}$ & $[12,8,2]_4$ & 4 & Theorem \ref{th.propagation1.direct sum} \\ 
			
			$\C$ & $\D$ & $[12,4,4]_4$ & 0(Hermitian LCD code) & Theorem \ref{th.propagation1.(u,u+v)-construction}.3) \\ 
			$\C^{\bot_H}$ & $\D$ & $[12,6,4]_4$ & 2 & Theorem \ref{th.propagation1.(u,u+v)-construction}.3) \\ 
			$\D$ & $\D$ & $[12,4,4]_4$ & 4(Hermitian self-orthogonal code) & Theorem \ref{th.propagation1.(u,u+v)-construction}.3) \\ 
			$\D^{\bot_H}$ & $\D$ & $[12,6,4]_4$ & 6(Hermitian self-dual code) & Theorem \ref{th.propagation1.(u,u+v)-construction}.3) \\ 
			\hline
		\end{tabular}
	\end{table}
\end{example}

\begin{remark}\label{rem5.mechanism of the two propagation rules}
	From Lemma \ref{lem.per-equvilent hull equal}, although some of the initial linear codes and resulting linear codes 
	in Table \ref{tab:mechanism of the propagation rules} have the same parameters, they are not permutation equivalent 
	because of the different dimensions of their Hermitian hulls, therefore all of the results 
	in Table \ref{tab:mechanism of the propagation rules} are dissimilar. 
	%\st{and these two propagation rules are well efficient for constructing new linear codes with prescribed hull dimensions.  }
\end{remark}

\subsection{Application to FSD LCD codes} 
%\st{In this subsection, we focus on the construction of new FSD LCD codes.  We} 
Note that the equivalence of LCD codes and linear codes under the Euclidean inner product 
and the Hermitian inner product introduced in \cite{CMTQ2018}, and hence, these facts 
motivate us to consider binary and ternary Euclidean LCD codes and quaternary Hermitian LCD codes. 

%\subsubsection{Constructions of New FSD LCD codes}
In the sequel, we first use the direct sum construction to derive new FSD codes from two given FSD 
codes.  And then we will provide some explicit constructions of new binary and ternary Euclidean FSD LCD codes 
as well as quaternary Hermitian FSD LCD codes. To this end, we need the following result. 

\begin{lemma}\label{lem.fsd.judge}{(Proposition 2.1 in \cite{L2022} and Proposition 1 in \cite{SXS2020})}.
	Let $A$ be an $n\times n$ matrix satisfying $A^T=QAQ$, where $Q$ is a monomial matrix such that $Q^2=I_n$ 
	and $I_n$ is an identity matrix of size $n\times n$. Then the linear code $\C$ generated by $(I_n\ A)$ is 
	a FSD code of length $2n$ with respect to the Euclidean inner product and the Hermitian inner product.  
\end{lemma}

\begin{theorem}\label{th.equvilent fsd code}
	Let $\C_1$, $\C_2$ be $[2n,n,d_1]_q$ and $[2n,n,d_2]_q$ linear codes. Let $G_1$ be a generator matrix of 
	$\C_1$ and $G_2$ be a generator matrix of $\C_2$, where $G_1=(I_{n}\ f_1(A_1))$, $G_2=(I_{n}\ f_2(A_2))$, 
	$f_i(x)\in \F_q[x]$ and $A_i$ is any Toeplitz matrix for $i=1,\ 2$. Then there exists a 
	$[4n,2n,\min\{d_1,d_2\}]_q$ FSD code $\C$ with respect to the Euclidean inner product and the Hermitian inner product. 
\end{theorem}
\begin{proof}
	Assume that $f_1(x)=\sum_{i=0}^{n-1}a_ix^i$ and $f_2(x)=\sum_{j=0}^{n-1}b_jx^j$ are two polynomials in $\F_q[x]$. 
	By \cite[Theorem 2.6]{L2022}, we know that there are monomial matrices $Q_1$ and $Q_2$ with $Q_1^2=Q_2^2=I_{n}$ 
	such that 
	\begin{equation}\label{eq.equvialent fsd 1}
	f_1(A_1)^T=Q_1f_1(A_1)Q_1\ {\rm and}\ f_2(A_2)^T=Q_2f_2(A_2)Q_2.
	\end{equation}
	
	By Theorem \ref{th.propagation1.direct sum}.{1)}, we know that $\C_1\oplus \C_2$ is a 
	$[4n,2n,\min\{d_1,d_2\}]_q$ linear code with a generator matrix 
	$$G=\left(\begin{array}[]{c|c}
	I_{n}\ f_1(A_1) & O_{n\times 2n} \\ \hline
	O_{n\times 2n} & I_{n}\ f_2(A_2) \\
	\end{array}\right).$$ 
	Since the size of $f_1(A_1)$ is $n\times n$, there exists a linear code $\C$ generated 
	by the matrix  
	$$G'=\left(\begin{array}[]{cc|cc}
	I_{n} & O_{n} & f_1(A_1) & O_{n} \\ 
	O_{n} & I_{n} &  O_{n} & f_2(A_2) \\
	\end{array}\right),$$  which is permutation equivalent 
	to $\C_1\oplus \C_2$. 
	
	If $F(A)=
	\left(\begin{array}[]{cc}
	f_1(A_1) & O_{n} \\  
	O_{n} & f_2(A_2) \\
	\end{array}\right)$, then $G'$ can be written as $G'=\left(I_{2n}\ F(A) \right)$. 
	Set $Q=\left(\begin{array}[]{cc}
	Q_1 & O_n \\
	O_n & Q_2 \\
	\end{array}\right)$, where $Q_1$ and $Q_2$ are defined as Equation~(\ref{eq.equvialent fsd 1}). 
	Then, it is easy to verify that $Q$ is a monomial matrix with 
	$Q^2=\left(\begin{array}[]{cc}
	Q_1^2 & O_n \\
	O_n & Q_2^2 \\
	\end{array}\right)=I_{2n}$  and $$
	F(A)^T  = \left(\begin{array}[]{cc}
	f_1(A_1)^T & O_n \\
	O_n & f_2(A_2)^T \\
	\end{array}\right)  
	= \left(
	\begin{array}[]{cc}
	Q_1f_1(A_1)Q_1 & O_{n} \\  
	O_{n} & Q_2f_2(A_2)Q_2 \\
	\end{array}
	\right)  
	= QF(A)Q.
	$$
	Therefore, by Lemma \ref{lem.fsd.judge}, we can conclude that $\C$ is a FSD code. 
\end{proof}
\begin{remark}\label{rem.equvialent fsd} $\,$
	\begin{enumerate}
		\item From \cite[Theorem 2.6]{L2022} , it is clear that Toeplitz matrices 
		$A_1$ and $A_2$ are essential to make Equation~(\ref{eq.equvialent fsd 1}) hold. 
		Hence, the matrices $A_1$ and $A_2$ in Theorem \ref{th.equvilent fsd code} can not be taken arbitrarily. 
		\item Keeping conditions of Theorem \ref{th.equvilent fsd code} and repeating the steps of 
		Theorem \ref{th.equvilent fsd code}, we can easily know that a $[4sn,2sn,\min\{d_1,d_2\}]_q$ 
		FSD code exists for any integer $s\geq 1$. 
	\end{enumerate}
\end{remark}

%\st{We now can derive FSD LCD codes in an explicit way.} 

\begin{theorem}\label{th.FSD LCD}
	With the same notation as in Theorem \ref{th.equvilent fsd code}, assume that both 
	$\C_1$ and $\C_2$ are Euclidean (resp. Hermitian) LCD codes. 
	Then there exists a $[4sn,2sn,\min\{d_1,d_2\}]_q$ Euclidean (resp. Hermitian) 
	FSD LCD code for any integer $s\geq 1$. 
\end{theorem}
\begin{proof}
	According to Remark \ref{rem.equvialent fsd}.{2)}, it is sufficient to prove 
	that the desired result holds for $s=1$. 
	When $s=1$, since both $\C_1$ and $\C_2$ are Hermitian LCD codes, then by Corollary \ref{coro1.th.propagation1}.{2)}, we know that $\C_1\oplus \C_2$ is a 
	$[4n,2n,\min\{d_1,d_2\}]$ Hermitian LCD code, which yields $\dim(\Hull_H(\C_1\oplus \C_2))=0$. 
	And by Theorem \ref{th.equvilent fsd code}, we can obtain a permutation equivalent $[4n,2n,\min\{d_1,d_2\}]$ FSD code. 
	Therefore, the desired result follows from Lemma \ref{lem.per-equvilent hull equal}{.2)}. 
\end{proof}

\begin{theorem}\label{th.2-ary FSD LCD codes}
	For $2n_1\in \{24, 26, 28, 30, 32\}$, $2n_2\in \{34, 36, 38\}$, there is a binary Euclidean FSD LCD 
	$[4sn_1,2sn_1,6]_2$ and a $[4sn_2,2sn_2,7]_2$ code for any integer $s\geq 1$. 
	In addition, the binary Euclidean FSD LCD $[44s,22s,5]_2$, $[80s,40s,8]_2$ and 
	$[100s,50s,9]_2$ codes exist for any integer $s\geq 1$. 
\end{theorem}
\begin{proof}
	In \cite[Table 1]{LS2022}, the authors list the known binary Euclidean FSD LCD codes obtained from 
	some Toeplitz matrices. Let $\C_1=\C_2=[50,25,9]_2$ be the binary Euclidean FSD LCD code 
	in \cite[Table 1]{LS2022}. From Theorem \ref{th.FSD LCD}, the binary Euclidean FSD LCD 
	$[100s,50s,9]_2$ code exists for any integer $s\geq 1$. Similarly, the remaining binary Euclidean FSD LCD codes exist. 
\end{proof}

\begin{theorem}\label{th.3-ary FSD LCD codes}
	For $2n_1\in \{18, 20\}$, $2n_2\in \{22, 24, 26\}$, $2n_3\in \{28, 30\}$, 
	there is a ternary Euclidean FSD LCD 
	$[4sn_1,2sn_1,6]_3$, $[4sn_2,2sn_2,7]_3$ and $[4sn_3,2sn_3,8]_3$ code for any integer 
	$s\geq 1$. In addition, the ternary Euclidean FSD LCD $[32s,16s,5]_3$ code exists for 
	any integer $s\geq 1$. 
\end{theorem}
\begin{proof}
	In \cite[Table 3]{LS2022}, the authors list the known ternary Euclidean FSD LCD codes 
	obtained from some Toeplitz matrices. Let $\C_1=\C_2=[16,8,5]_3$ be the ternary Euclidean 
	FSD LCD code in \cite[Table 3]{LS2022}. From Theorem \ref{th.FSD LCD}, the ternary 
	Euclidean FSD LCD $[32s,16s,5]_3$ code exists for any integer $s\geq 1$. 
	Similarly, the remaining ternary Euclidean FSD LCD codes exist. 
\end{proof}

\begin{theorem}\label{th.4-ary FSD LCD codes}
	Quaternary Hermitian FSD LCD $[32s,16s,6]_4$, $[40s,20s,7]_4$ and $[48s,24s,8]_4$ codes 
	exist for any integer $s\geq 1$. 
\end{theorem}
\begin{proof}
	In \cite[Table 5]{LS2022}, the authors list the known quaternary Hermitian FSD LCD codes 
	obtained from some Toeplitz matrices. Let $\C_1=\C_2=[24,12,8]_4$ be the quaternary Hermitian 
	FSD LCD code in \cite[Table 5]{LS2022}. From Theorem \ref{th.FSD LCD}, the quaternary  
	Hermitian FSD LCD $[48s,24s,8]_4$ code exists for any integer $s\geq 1$. 
	Similarly, the remaining quaternary Hermitian FSD LCD codes exist. 
\end{proof}

\begin{remark}\label{rem.fsd lcd codes}
	As it was claimed in the proof of Theorems \ref{th.2-ary FSD LCD codes}-\ref{th.4-ary FSD LCD codes}, 
	in \cite{LS2022}, the authors chose some suitable Toeplitz matrices and polynomials 
	$f(x)$ to generate good FSD LCD codes based on computer searches. However, we obtain 
	many FSD LCD codes using the direct sum construction in Theorems 
	\ref{th.2-ary FSD LCD codes}-\ref{th.4-ary FSD LCD codes}. 
	In comparison with \cite{LS2022}, it is apparent that our computational complexity 
	is much less. In fact, the reader  can notice that our method needs little additional 
	computation. Naturally, this phenomenon will be more noticeable for the construction 
	of FSD LCD codes of large lengths. Furthermore, we suspect that based on the direct 
	sum construction, one could possibly design more efficient algorithms for finding 
	appropriate Toeplitz matrices and polynomials $f(x)$ to construct FSD LCD codes with 
	larger minimum distances. 
\end{remark}

\subsection{Application to even-like codes}
Note that Corollary \ref{coro2.th.propagation2} depends on the existence of even-like 
codes. %\st{it is necessary to give some known and interesting} 
Therefore we will give some even-like codes so that we can use Corollary 
\ref{coro2.th.propagation2} to construct binary Euclidean (almost) self-orthogonal codes. 
%\st{We rephase some of them in the following proposition.} 

\begin{proposition}[\cite{HKS2021}]\label{prop.even weight codes}
	$\C$ is an even-like code if one of the following conditions holds:   
	\begin{enumerate}
		\item  $\C$ is an $[n,n-1,2]_2$ code, where $n\geq 2$;
		\item  $\C$ is a binary Euclidean self-orthogonal code; 
		\item  $\C$ is a binary simplex $[2^t-1,t,2^{t-1}]_2$ code. 
		%\item [\rm (3)] $\C$ is a quaternary Hermitian self-orthogonal code; 
	\end{enumerate}
\end{proposition}

In addition to the even-like codes listed in Proposition \ref{prop.even weight codes}, 
the following two theorems give two general methods for constructing even-like and odd-like codes. 
They come from the applications of two propagation rules. The first theorem is related to the 
$(\mathbf{u},\mathbf{u+v})$ construction. 

\begin{theorem}\label{th.even-like and odd-like codes from uv}
	%Let $\C_2$ be an $[n,1,n]_q$ repetition code generated by $\mathbf{1_n}=(\underbrace{1\ 1\ \cdots\ 1}_{n})$ and other notations be the same as Theorem \ref{th.propagation1.(u,u+v)-construction}. If $\C_1$ is an even weight code, then the following statements hold. 
	Let $\C_2$ be an $[n,1,n]_2$ repetition code generated by $\mathbf{1_n}=(\underbrace{1\ 1\ \cdots\ 1}_{n})$ 
	and other notations be the same as Theorem \ref{th.propagation1.(u,u+v)-construction}. Then the following statements hold. 
	\begin{enumerate}
		\item  $\C$ is a $[2n,k_1+1,\min\{2d_1,n\}]_2$ odd-like code if and only if $n$ is odd.  
		
		\item  $\C$ is a $[2n,k_1+1,\min\{2d_1,n\}]_2$ even-like code if and only if $n$ is even. 
	\end{enumerate}
\end{theorem}
\begin{proof}
	From Theorem \ref{th.propagation1.(u,u+v)-construction}.{1)}, under the given conditions, 
	$\C$ is an $[2n,k_1+1,\min\{2d_1,n\}]_2$ linear code with a generator matrix 
	$$G=\left( \begin{array}[]{cc}
	G_1 & G_1 \\
	O_{1\times n} & \mathbf{1_n} \\
	\end{array} \right).$$ 
	Set $$G_1=\left( \begin{array}{cccc}
	g_{11} & g_{12} & \cdots & g_{1n} \\
	g_{21} & g_{22} & \cdots & g_{2n} \\
	\vdots & \vdots & \cdots & \vdots \\ 
	g_{k_11} & g_{k_12} & \cdots & g_{k_1 n} \\
	\end{array}
	\right)$$ 
	%Denote $(\underbrace{g_{i1}\ g_{i2}\ \cdots\ g_{in}}_{n}\ \underbrace{g_{i1}\ g_{i2}\ \cdots\ g_{in}}_{n})$ by $\mathbf{g_{i}}$ for $1\leq i\leq k_1$ and denote $(\underbrace{0\ 0\ \cdots\ 0}_{n}\ \underbrace{1\ 1\ \cdots\ 1}_{n})$ by $\mathbf{g_{k_1+1}}$. Then, $G$ can be written as $$G=(\mathbf{g_1}^T\ \mathbf{g_2}^T\ \cdots\ \mathbf{g_{k_1}}^T\ \mathbf{g_{k_1+1}}^T)^T,$$ 
	and hence, any one codeword $\bc$ of $\C$ can be expressed as 
	$$\bc=(\underbrace{\sum_{i=1}^{k_1}u_ig_{i1}, \sum_{i=1}^{k_1}u_ig_{i2}, \cdots, \sum_{i=1}^{k_1}u_ig_{in}}_{n},
	\underbrace{u_{k_1+1}+\sum_{i=1}^{k_1}u_ig_{i1}, u_{k_1+1}+\sum_{i=1}^{k_1}u_ig_{i2}, \cdots, u_{k_1+1}+\sum_{i=1}^{k_1}u_ig_{in}}_{n}),$$ 
	where $u_1,u_2,\dots,u_{k_1},u_{k_1+1}\in \F_2$. Since $q=2$, it follows that the weight of any codeword $\bc\in \C$ can be 
	calculated as follows: 
\begin{equation*}
	\begin{split}
	\wt(\bc) & = \sum_{i=1}^{k_1}u_ig_{i1}+\sum_{i=1}^{k_1}u_ig_{i2}+\cdots+\sum_{i=1}^{k_1}u_ig_{in}
	+u_{k_1+1}+\sum_{i=1}^{k_1}u_ig_{i1}+ u_{k_1+1}+\sum_{i=1}^{k_1}u_ig_{i2}+ \cdots+ u_{k_1+1}+\sum_{i=1}^{k_1}u_ig_{in} \\
	& = 2(\sum_{i=1}^{k_1}u_ig_{i1}+\sum_{i=1}^{k_1}u_ig_{i2}+\cdots+\sum_{i=1}^{k_1}u_ig_{in})+n\cdot u_{k_1+1}.
	\end{split}
\end{equation*}
	Noting that $2(\sum_{i=1}^{k_1}u_ig_{i1}+\sum_{i=1}^{k_1}u_ig_{i2}+\cdots+\sum_{i=1}^{k_1}u_ig_{in})$ 
	must be even, we have $\wt(\bc)$ is even for any $\bc \in \C$ if and only if $n$ is even. 
	This completes the proof.  
\end{proof}
\begin{remark}
	Note that in Theorem 
	\ref{th.even-like and odd-like codes from uv}, the initial code $C_1$ 
	need not be even-like. 
\end{remark}

From Proposition \ref{prop.even weight codes}.{2)}, we know that binary Euclidean self-orthogonal 
codes are even-like. However, there are even-like codes that are not binary Euclidean 
self-orthogonal. From Corollary \ref{coro2.th.propagation2} and 
Theorem \ref{th.even-like and odd-like codes from uv}, we can unify both of them in our construction, i.e., we can 
get the following corollary. 
\begin{corollary}
	{With the  notations and conditions stated in Theorem \ref{th.propagation2.(u,u+v)-construction}}. 
	Then the following statements hold. 
	\begin{enumerate}
		\item $\C$ is a $[2n,k_1+1,\min\{2d_1,n\}]_2$ odd-like code if and only if $\C$ is 
		a binary Euclidean almost self-orthogonal code.  
		
		\item $\C$ is a $[2n,k_1+1,\min\{2d_1,n\}]_2$ even-like code if and only if $\C$ is 
		a binary Euclidean self-orthogonal code. 
	\end{enumerate}
\end{corollary}
\begin{proof}
	The desired results follow immediately from 
	Corollary \ref{coro2.th.propagation2} and 
	Theorem \ref{th.even-like and odd-like codes from uv}. 
\end{proof}

Similarly, by employing the direct sum construction over $\F_2$, one can also easily determine 
whether the derived code is even-like or odd-like. The proof is completely similar to Theorem
\ref{th.even-like and odd-like codes from uv}, so it is omitted. 

\begin{theorem}\label{th.even-like and odd-like codes from direct sum}
	{With the  notations and conditions stated in Theorem \ref{th.propagation1.direct sum}.} Then the following statements hold over $\F_2$. 
	\begin{enumerate}
		\item  If both $\C_1$ and $\C_2$ are even-like, then $\C_1\oplus \C_2$ is even-like.
		\item  If one of $\C_1$ and $\C_2$ is odd-like, then $\C_1\oplus \C_2$ is odd-like. 
		\item  If both $\C_1$ and $\C_2$ are odd-like, then $\C_1\oplus \C_2$ is even-like or odd-like.  
	\end{enumerate}
\end{theorem}

\subsection{Applications to binary (almost) self-orthogonal codes}
In this subsection, we employ the results described in Corollary \ref{coro2.th.propagation2} to 
construct some interesting binary Euclidean (almost) self-orthogonal codes. Many of them 
are optimal according to Database \cite{Grass} and are self-dual. 
In particular, a family of constructed binary almost self-orthogonal codes 
meets {the so-called Griesmer's bound} given in Equation~(\ref{eq. Griesmer bound}). 

\begin{theorem}\label{th.binary sd with d=4}
	The following statements hold.
	\begin{enumerate}
		\item A binary Euclidean almost self-dual $[2n,n,\min\{4,n\}]_2$ code exists for each 
		odd $n\geq 3$. Moreover, the binary Euclidean almost self-dual $[2n,n,4]_2$ code exists for each odd $n\geq 5$. 
		
		\item A binary Euclidean self-dual $[2n,n,\min\{4,n\}]_2$ code exists for each even $n\geq 2$. 
		Moreover, the binary Euclidean self-dual $[2n,n,4]_2$ code exists for each even $n\geq 4$. 
	\end{enumerate}
\end{theorem}
\begin{proof}
	Let $\C_1$ be the $[n,n-1,2]_2$ code and $\C_2$ be the $[n,1,n]_2$ code, where $n\geq 2$. 
	From Proposition \ref{prop.even weight codes}.{1)}, $\C_1$ is an even-like code. 
	Then the desired results follow immediately from Corollary \ref{coro2.th.propagation2}
	and the definition of binary Euclidean (almost) self-dual codes. 
\end{proof}

\begin{example}\label{exam.binary sd with d=4}
	Taking $2\leq n\leq 10$, we list all possible binary Euclidean (almost) self-dual 
	codes derived from Theorem \ref{th.binary sd with d=4} in Table \ref{tab:binary sd}. 
	Notice that most of them are (almost) optimal according to Database \cite{Grass}. 
	
	\begin{table}[!htb]
		\centering
		\caption{Some binary (almost) self-dual codes derived from Theorem \ref{th.binary sd with d=4} for $2\leq n\leq 10$}
		\label{tab:binary sd}       % Give a unique label
		\begin{tabular}{cccc}\hline
			$n$ & $[n,k,d]_2$ &{\bf Self-dual/Almost self-dual} & {\bf Optimal/Almost Optimal} \\ \hline
			2 & $[4,2,2]_2$ & Self-dual & Optimal \\
			3 & $[6,3,3]_2$ & Almost self-dual & Optimal \\
			4 & $[8,4,4]_2$ & Self-dual & Optimal \\
			5 & $[10,5,4]_2$ & Almost self-dual & Optimal \\
			
			6 & $[12,6,4]_2$ & Self-dual & Optimal \\
			7 & $[14,7,4]_2$ & Almost self-dual & Optimal \\
			8 & $[16,8,4]_2$ & Self-dual & Almost optimal \\
			9 & $[18,9,4]_2$ & Almost self-dual & - \\
			10 & $[20,10,4]_2$ & Self-dual & - \\
			
			\hline
		\end{tabular}
	\end{table}
	
\end{example}

\begin{theorem}\label{th.binary so with d=2^{t+2}}
	A binary even-like Euclidean self-orthogonal $[2^{t+1}n,n+t,2^{t+2}]_2$ code exists for each integer $t\geq 0$ and 
	even $n\geq 4$. 
\end{theorem}
\begin{proof}
	From Theorem \ref{th.binary sd with d=4}, the binary self-dual $[2n,n,4]_2$ code exists 
	for each even $n\geq 4$. And from Proposition \ref{prop.even weight codes}.{2)}, for each 
	even $n\geq 4$, the $[2n,n,4]_2$ code is even-like. Then, the desired result can be derived 
	from the following process: 
	\begin{itemize}
		\item Take the even-like $[2n,n,4]_2$ code as $\C_1$ and take the repetition 
		$[2n,1,2n]_2$ code as $\C_2$. Since $n\geq 4$ is even, we can derive even-like 
		$[4n,n+1,8]_2$ codes from Theorem \ref{th.even-like and odd-like codes from uv}. 
		It follows from Corollary \ref{coro2.th.propagation2} that 
		such $[4n,n+1,8]_2$ codes are also binary Euclidean self-orthogonal codes. 
		
		\item Take the even-like $[4n,n+1,8]_2$ code as $\C_1$ and take the repetition 
		$[4n,1,4n]_2$ code as $\C_2$. Since $n\geq 4$ is even, we can derive even-like 
		$[8n,n+2,16]_2$ codes from Theorem \ref{th.even-like and odd-like codes from uv}. 
		It follows from Corollary \ref{coro2.th.propagation2} that 
		such $[8n,n+2,16]_2$ codes are also binary Euclidean self-orthogonal codes. 
		
		\item And we keep on this process.
		
		\item Take the even-like $[2^tn,n+t-1,2^{t+1}]_2$ code as $\C_1$ and take the repetition 
		$[2^tn,1,2^tn]_2$ code as $\C_2$. Since $n\geq 4$ is even, we can derive even-like 
		$[2^{t+1}n,n+t,2^{t+2}]_2$ codes from Theorem \ref{th.even-like and odd-like codes from uv}. 
		It follows from Corollary \ref{coro2.th.propagation2} that 
		such $[2^{t+1}n,n+t,2^{t+2}]_2$ codes are also binary Euclidean self-orthogonal codes. 
	\end{itemize}
	
	Notice that $[2^{t+1}n,n+t,2^{t+2}]_2$ codes are still even-like, then the above procedure 
	can continue to be executed, and hence, { a} binary Euclidean self-orthogonal 
	$[2^{t+1}n,n+t,2^{t+2}]_2$ code exists for each integer $t\geq 0$ and even $n\geq 4$.  
\end{proof}

\begin{example}\label{exam.binary so with d=2^{t+2}}
	Taking $0\leq t\leq 5$ and even $4\leq n\leq 8$, we list all possible binary even-like self-orthogonal 
	codes derived from Theorem \ref{th.binary so with d=2^{t+2}} in Table \ref{tab:binary SO}. 
	Notice that many of them are (almost) optimal according to Database \cite{Grass}.

	\begin{table}[!htb]
		\centering
		\caption{All possible binary self-orthogonal code derived from Theorem 
			\ref{th.binary so with d=2^{t+2}} for $0\leq t\leq 5$ and even $4\leq n\leq 8$}
		\label{tab:binary SO}       % Give a unique label
		\begin{tabular}{cccc}\hline
			 $t$ & $n$ & {\bf Binary\ self-orthogonal\ codes} & {\bf Optimal/Almost Optimal} \\ \hline
			0 & 4 & $[8,4,4]_2$ & Optimal \\
			0 & 6 & $[12,6,4]_2$ & Optimal \\
			1 & 4 & $[16,5,8]_2$ & Optimal \\
			0 & 8 & $[16,8,4]_2$ & Almost optimal \\
			
			1 & 6 & $[24,7,8]_2$ & - \\
			2 & 4 & $[32,6,16]_2$ & Optimal \\
			1 & 8 & $[32,9,8]_2$ & - \\
			
			2 & 6 & $[48,8,16]_2$ & - \\
			3 & 4 & $[64,7,32]_2$ & Optimal \\
			
			2 & 8 & $[64,10,16]_2$ & - \\ 
			3 & 6 & $[96,9,32]_2$ & - \\
			4 & 4 & $[128,8,64]_2$ & Optimal \\ 
			
			3 & 8 & $[128,11,32]_2$ & - \\
			4 & 6 & $[192,10,64]_2$ & - \\
			5 & 4 & $[256,9,128]_2$ & Optimal \\
			
			4 & 8 & $[256,12,64]_2$ & - \\
			5 & 6 & $[384,11,128]_2$ & - \\
			5 & 8 & $[512,13,128]_2$ & - \\
			\hline
		\end{tabular}
	\end{table}
\end{example}

\begin{theorem}\label{th.binary so with d=2^t}
	A binary even-like Euclidean self-orthogonal $[2^{t+1},t+1,2^{t}]_2$ code exists for 
	each integer $t\geq 1$. % and such codes are optimal. 
\end{theorem}
\begin{proof}
	From Theorem \ref{th.binary sd with d=4}.{1)}, there is a binary self-dual $[4,2,2]_2$ code, 
	which is even-like according to Proposition \ref{prop.even weight codes}.{2)}. Taking 
	the $[4,2,2]_2$ code as $\C_1$ and the $[4,1,4]_2$ code as $\C_2$, a binary Euclidean 
	self-orthogonal $[8,3,4]_2$ code can be derived from Corollary \ref{coro2.th.propagation2}, 
	which is also even-like according to Proposition \ref{prop.even weight codes}.{2)} or Theorem \ref{th.even-like and odd-like codes from uv}.{2)}. 
	Analogous to the proof of Theorem \ref{th.binary so with d=2^{t+2}}, we can obtain 
	a binary Euclidean self-orthogonal $[2^{t+1},t+1,2^{t}]_2$ code for each integer 
	$t\geq 1$. 
	
	%We now proof that such codes are optimal. 
\end{proof} 
\begin{example}\label{exam.binary so with d=2^t}
	According to Database \cite{Grass}, 
	{the binary even-like Euclidean self-orthogonal $[2^{t+1},t+1,2^{t}]_2$ code derived 
		from Theorem \ref{th.binary so with d=2^t} is optimal for each $1\leq t\leq 7$.} 
	All parameters are listed as follows: 
	$$[4,2,2]_2,\ [8,3,4]_2,\ [16,4,8]_2,\ $$
	$$[32,5,16]_2,\ [64,6,32]_2,\ [128,7,64]_2,\ [256,8,128]_2.$$
\end{example}

\begin{theorem}\label{th.binary almost so with d=2^k-1}
	A binary odd-like Euclidean almost self-orthogonal $[2^{t+1}-2,t+1,2^{t}-1]_2$ code exists for 
	each integer $t\geq 1$ and such codes meet {the Griesmer's bound}. 
\end{theorem}
\begin{proof}
	From Proposition \ref{prop.even weight codes}.{3)}, the binary simplex $[2^t-1,t,2^{t-1}]_2$ 
	code is even-like for each integer $t\geq 1$. Taking the $[2^t-1,t,2^{t-1}]_2$ code as 
	$\C_1$ and the $[2^t-1,1,2^t-1]_2$ code as $\C_2$, a binary Euclidean almost 
	self-orthogonal $[2^{t+1}-2,t+1,2^{t}-1]_2$ code can be derived from 
	Corollary \ref{coro2.th.propagation2} and the fact that $2^{t}-1$ is odd 
	for each integer $t\geq 1$, which is also odd-like according to 
	Theorem \ref{th.even-like and odd-like codes from uv}.{1)} . 
	
	We now prove that such codes meet  the Griesmer's bound. For $q=2$, 
	from Equation~(\ref{eq. Griesmer bound}), we have 
	$$\sum_{i=0}^{t+1-1} \lceil \frac{2^t-1}{2^i} \rceil=2^t-1+2^{t-1}+\cdots+2+1=2^{t+1}-2,$$
	which completes the proof. 
\end{proof}

\begin{example}\label{exam.binary almost so with d=2^t-1}
	According to Theorem \ref{th.binary almost so with d=2^k-1}, {the obtained binary odd-like Euclidean 
		almost self-orthogonal $[2^{t+1}-2,t+1,2^{t}-1]_2$ code} is optimal for each integer 
	$t\geq 1$. For $1\leq t\leq 10$, all parameters are listed as follows: 
	$$[2,2,1]_2,\ [6,3,3]_2,\ [14,4,7]_2,\ [30,5,15]_2,\ [62,6,31]_2,$$
	$$[126,7,63]_2,\ [254,8,127]_2,\ [510,9,255]_2,\ [1022,10,511]_2,\ [2046,11,1023]_2$$
\end{example}

\subsection{Application to non-binary self-orthogonal codes}
In this subsection, we employ Corollary \ref{coro1.th.propagation2} and nested GRS codes 
to obtain some interesting Euclidean self-orthogonal codes over $\F_4$ and $\F_8$, and 
Hermitian self-orthogonal codes over $\F_{16}$ and $\F_{64}$.
Noting that some of them are (almost) optimal codes according to Database \cite{Grass}. 
To this end, we need the following results. 

\begin{lemma}\label{lem.Euclidean So GRS codes}%{\rm \cite[Theorem 3]{LCC2018}}
	Let $q=2^m\geq 4$. If $3<n\leq q$, there exists an Euclidean self-orthogonal 
	$[n,k,n-k+1]_q$ GRS code for any $1\leq k\leq \lfloor \frac{n}{2} \rfloor$.  
\end{lemma}
\begin{proof}
	From the proof of \cite[Theorem 3]{LCC2018}, we can know that there exists 
	an $[n,k,n-k+1]_q$ GRS code with $l$-dimensional Euclidean hull for any 
	$1\leq l\leq k$ under given conditions. Taking $l=k$, the desired result follows. 
\end{proof}

\begin{lemma}\label{lem.Hermitian So GRS codes}%{\rm \cite[Theorem 3.5]{FFLZ2019}}
	Let $q=2^m\geq 4$. If $n\leq q$, there exists a Hermitian self-orthogonal 
	$[n,k,n-k+1]_{q^2}$ GRS code for any $1\leq k\leq \lfloor \frac{n}{2} \rfloor$.  
\end{lemma}
\begin{proof}
	From the proof of \cite[Theorem 3.5]{FFLZ2019}, we can know that there exists 
	an $[n,k,n-k+1]_{q^2}$ GRS code with $l$-dimensional Euclidean hull for any 
	$0\leq l\leq k$ under given conditions. Taking $l=k$, the desired result follows. 
\end{proof}

\begin{lemma}\label{lem.Hermitian So GRS codes2}%{\rm \cite[Theorem 3.5]{FFLZ2019}}
	Let $q=2^m\geq 4$. There exists a Hermitian self-orthogonal 
	$[q^2+1,q,q^2-q+2]_{q^2}$ GRS code.  
\end{lemma}
\begin{proof}
	From the proof of \cite[Theorem 3.11]{FFLZ2019}, we can know that there exists 
	an $[q^2+1,q,q^2-q+2]_{q^2}$ GRS code with $l$-dimensional Euclidean hull for any 
	$0\leq l\leq q$ under given conditions. Taking $l=q$, the desired result follows. 
\end{proof}

\begin{theorem}\label{th.quaternary and octet Euclidean SO codes}
	The quaternary and octet Euclidean self-orthogonal or self-dual codes listed in Table 
	\ref{tab:quaternary and octet Euclidean self-orthogonal codes} exist. 
\end{theorem}
\begin{proof}
	From Lemma \ref{lem.Euclidean So GRS codes}, we know that Euclidean self-orthogonal 
	$[4,1,4]_4$, $[4,1,4]_8$, $[5,1,5]_8$, $[5,2,4]_8$, $[6,1,6]_8$, $[6,2,5]_8$, 
	$[7,1,7]_8$, $[7,2,6]_8$, $[8,1,8]_8$ and $[8,2,7]_8$ GRS codes exist. 
	Taking the Euclidean self-orthogonal $[4,1,4]_4$ GRS code as an example, the desired 
	result can be derived from the following steps: 
	\begin{itemize}
		\item \textbf{Step 1:} From Lemma \ref{lem.GRS.Euclidean}, the Euclidean dual 
		code of the $[4,1,4]_4$ GRS code is a $[4,3,2]_4$ GRS code; 
		\item \textbf{Step 2:} From Proposition \ref{prop.GRS nested}, for any $1\leq k\leq 3$, 
		there exists a $[4,k,5-k]_4$ GRS code, which is contained to the $[4,3,2]_4$ GRS code;
		\item \textbf{Step 3:} Taking the $[4,k,5-k]_4$ GRS code as $\C_1$ and the Euclidean 
		self-orthogonal $[4,1,4]_4$ GRS code as $\C_2$, it follows from 
		Corollary \ref{coro1.th.propagation2}  that an Euclidean self-orthogonal 
		or self-dual $[8,k+1,\min\{2\times(5-k), 4\}]_8$ code exists for any $1\leq k\leq 3$. 
	\end{itemize} 
	
	\begin{table}[h]
		\centering
		\caption{Some interesting quaternary and octet Euclidean self-orthogonal code from Theorem \ref{th.quaternary and octet Euclidean SO codes}}
		\label{tab:quaternary and octet Euclidean self-orthogonal codes}       % Give a unique label
		\begin{tabular}{cccccc}\hline
			$q$ & $\C_1$ & $\C_2$ & \tabincell{c}{New Euclidean\\self-orthogonal codes} & \tabincell{c}{Optimal minimum \\ distance in \cite{Grass}} & $\tabincell{c}{Optimal/\\ Almost optimal }$ \\ \hline
			$q=4$ & $[4, 1, 4]_4$ & $[4, 1, 4]_4$ & $[8, 2, 4]_4$ & 6  & -\\ 
			& $[4, 2, 3]_4$ & $[4, 1, 4]_4$ & $[8, 3, 4]_4$ & 5  & Almost optimal\\ 
			& $[4, 3, 2]_4$ & $[4, 1, 4]_4$ & $[8, 4, 4]_4$ & 4  & Optimal \\ \hline
			
			$q=8$ & $[4, 2, 3]_8$ & $[4, 1, 4]_8$ & $[8, 3, 4]_8$ & 6  & -\\ 
			& $[4, 3, 2]_8$ & $[4, 1, 4]_8$ & $[8, 4, 4]_8$ & 5  & Almost optimal \\

			& $[5, 3, 3]_8$ & $[5, 1, 5]_8$ & $[10, 4, 5]_8$ & 6 & Almost optimal \\
			& $[5, 3, 3]_8$ & $[5, 2, 4]_8$ & $[10, 5, 4]_8$ & 5 & Almost optimal \\
			
			& $[6, 3, 4]_8$ & $[6, 1, 6]_8$ & $[12, 4, 6]_8$ & 8 & - \\
			& $[6, 4, 3]_8$ & $[6, 1, 6]_8$ & $[12, 5, 6]_8$ & 7 & Almost optimal \\
			& $[6, 4, 3]_8$ & $[6, 2, 5]_8$ & $[12, 6, 5]_8$ & 6 & Almost optimal \\
			
			& $[7, 4, 4]_8$ & $[7, 1, 7]_8$ & $[14, 5, 7]_8$ & 9 & - \\
			& $[7, 5, 3]_8$ & $[7, 1, 7]_8$ & $[14, 6, 6]_8$ & 8 & - \\
			& $[7, 5, 3]_8$ & $[7, 2, 6]_8$ & $[14, 7, 6]_8$ & 7 & Almost optimal \\
			
			& $[8, 4, 5]_8$ & $[8, 1, 8]_8$ & $[16, 5, 8]_8$ & 10 & - \\ 
			& $[8, 5, 4]_8$ & $[8, 1, 8]_8$ & $[16, 6, 8]_8$ & 9 & Almost optimal \\ 
			& $[8, 6, 3]_8$ & $[8, 1, 8]_8$ & $[16, 7, 6]_8$ & 8 & - \\ 
			& $[8, 6, 3]_8$ & $[8, 2, 7]_8$ & $[16, 8, 6]_8$ & 8 & - \\ 
			\hline
		\end{tabular}
	\end{table}
\end{proof}

\begin{theorem}\label{th.Hermitian SO codes over F4, F16, F64}
	The Hermitian self-orthogonal or self-dual codes over $\F_{16}$ and $\F_{64}$ listed 
	in Table \ref{tab:Hermitian SO codes over F4, F16, F64} exist. 
\end{theorem}
\begin{proof}
	From Lemmas \ref{lem.Hermitian So GRS codes} and \ref{lem.Hermitian So GRS codes2}, 
	we know that Hermitian self-orthogonal 
	$[2, 1, 2]_{16}$, $[3, 1, 3]_{16}$, $[4, 1, 4]_{16}$, $[17, 4, 14]_{16}$, 
	$[2, 1, 2]_{64}$, $[3, 1, 3]_{64}$, $[4, 1, 4]_{64}$, $[5, 1, 5]_{64}$, $[5, 2, 4]_{64}$, 
	$[6, 1, 6]_{64}$, $[6, 2, 5]_{64}$, $[7, 1, 7]_{64}$, $[7, 2, 6]_{64}$, $[8, 1, 8]_{64}$, 
	$[8, 2, 7]_{64}$ and $[65, 8, 58]_{64}$
	GRS codes exist. With the same argument of the proof of Theorem 
	\ref{th.quaternary and octet Euclidean SO codes}, the desired result follows. 
	
	\begin{table}[h]
		\centering
		\caption{Some interesting Hermitian self-orthogonal code over $\F_{16}$ and $\F_{64}$ from Theorem \ref{th.Hermitian SO codes over F4, F16, F64}}
		\label{tab:Hermitian SO codes over F4, F16, F64}       % Give a unique label
		\begin{tabular}{cccc}\hline
			$q^2$ & $\C_1$ & $\C_2$ & \tabincell{c}{New Hermitian self-orthogonal codes} \\ \hline
			%$q^2=4$ & $[2, 1, 2]_4$ & $[2, 1, 2]_4$ & $[4, 2, 2]_4$ \\ \hline
			
			$q^2=16$ & $[2, 1, 2]_{16}$ & $[2, 1, 2]_{16}$ & $[4, 2, 2]_{16}$ \\ 
			& $[3, 2, 2]_{16}$ & $[3, 1, 3]_{16}$ & $[6, 3, 3]_{16}$ \\ 
			& $[4, 3, 2]_{16}$ & $[4, 1, 4]_{16}$ & $[8, 4, 4]_{16}$ \\ 
			& $[17, 11, 7]_{16}$ & $[17, 4, 14]_{16}$ & $[34, 15, 14]_{16}$ \\  
			& $[17, 12, 6]_{16}$ & $[17, 4, 14]_{16}$ & $[34, 16, 12]_{16}$ \\
			& $[17, 13, 5]_{16}$ & $[17, 4, 14]_{16}$ & $[34, 17, 10]_{16}$ \\ \hline
			
			$q^2=64$ & $[2, 1, 2]_{64}$ & $[2, 1, 2]_{64}$ & $[4, 2, 2]_{64}$ \\ 
			& $[3, 2, 2]_{64}$ & $[3, 1, 3]_{64}$ & $[6, 3, 3]_{64}$ \\ 
			& $[4, 3, 2]_{64}$ & $[4, 1, 4]_{64}$ & $[8, 4, 4]_{64}$ \\
			
			& $[5, 3, 3]_{64}$ & $[5, 1, 5]_{64}$ & $[10, 4, 5]_{64}$ \\
			& $[5, 3, 3]_{64}$ & $[5, 2, 4]_{64}$ & $[10, 5, 4]_{64}$ \\
			
			& $[6, 4, 3]_{64}$ & $[6, 1, 6]_{64}$ & $[12, 5, 6]_{64}$ \\
			& $[6, 4, 3]_{64}$ & $[6, 2, 5]_{64}$ & $[12, 6, 5]_{64}$ \\
			
			& $[7, 4, 4]_{64}$ & $[7, 1, 7]_{64}$ & $[14, 5, 7]_{64}$ \\
			%& $[7, 4, 4]_{64}$ & $[7, 2, 6]_{64}$ & $[14, 6, 6]_{64}$ \\
			& $[7, 5, 3]_{64}$ & $[7, 2, 6]_{64}$ & $[14, 7, 6]_{64}$ \\
			
			& $[8, 5, 4]_{64}$ & $[8, 1, 8]_{64}$ & $[16, 6, 8]_{64}$ \\
			& $[8, 5, 4]_{64}$ & $[8, 2, 7]_{64}$ & $[16, 7, 7]_{64}$ \\
			& $[8, 6, 3]_{64}$ & $[8, 2, 7]_{64}$ & $[16, 8, 6]_{64}$ \\ 
			
			& $[65, 37, 29]_{64}$ & $[65, 8, 58]_{64}$ & $[130, 45, 58]_{64}$ \\ 
			& $[65, 42, 24]_{64}$ & $[65, 8, 58]_{64}$ & $[130, 50, 48]_{64}$ \\ 
			& $[65, 47, 19]_{64}$ & $[65, 8, 58]_{64}$ & $[130, 55, 38]_{64}$ \\ 
			& $[65, 52, 14]_{64}$ & $[65, 8, 58]_{64}$ & $[130, 60, 28]_{64}$ \\ 
			& $[65, 57, 9]_{64}$ & $[65, 8, 58]_{64}$ & $[130, 65, 18]_{64}$ \\ 
			\hline
		\end{tabular}
	\end{table}
\end{proof}

\begin{remark}$\,$
	\begin{enumerate}
		\item As it can be checked in Table~\ref{tab:quaternary and octet Euclidean self-orthogonal codes}, 
		many (almost) optimal Euclidean self-orthogonal and self-dual codes can be derived from this method. 
		\item Note that Database \cite{Grass} does not contain linear codes over $\F_{16}$ and $\F_{64}$. 
		And yet based on the excellent performance of the Euclidean self-orthogonal or self-dual codes obtained 
		in Table \ref{tab:quaternary and octet Euclidean self-orthogonal codes}, it is reasonable to believe that 
		the Hermitian self-orthogonal or self-dual codes obtained in Table \ref{tab:Hermitian SO codes over F4, F16, F64} 
		will similarly contain some (almost) optimal codes. 
	\end{enumerate}
	
\end{remark}

\subsection{Application to improve bounds on LCD codes}

In many works related to LCD codes (see for example \cite{B2021,L2022,LS2022} 
and the references therein), inequalities are used to enhance the known lower bound on the 
minimum distance of a LCD code. For these inequalities to be valid, one has to know the 
lower bound of the minimum distance of a LCD code of a certain length. From Corollary 
\ref{coro1.th.propagation1}.{2)}, we know that our constructions can be based on known 
LCD codes of small lengths to obtain LCD codes of larger lengths. Hence, more LCD codes 
can be known, and thus, we believe that combining our conclusions with these inequalities 
can help to better and faster determine the lower bounds on the minimum distance for LCD 
codes of larger lengths. 

Additionally, our methods themselves can improve the lower bounds of the minimum 
distance of some known LCD codes. We notice that, in 2021, Bouyuklieva et al. listed bounds for 
the minimum distance of binary Euclidean LCD codes of length $16\leq n\leq 40$ and 
$5\leq k\leq 32$ in \cite[Tables 1 and 2]{B2021}. We employ our results described in 
Corollary \ref{coro1.th.propagation1} to \cite[Table 1]{B2021} to improve some of the lower 
bounds in \cite[Table 2]{B2021}. Partial specific examples are listed in Table \ref{tablast} 
below and one can find more examples in a similar way. 

\begin{table}[!htb]
	\centering
	\caption{Improved lower bounds on the minimum distance of some binary Euclidean LCD codes in \cite[Table 2]{B2021}}
	\label{tablast}       % Give a unique label
	\begin{tabular}{ccccc}\hline
		$\C_1$ & $\C_2$ &  Derived LCD codes & Best in \cite[Table 2]{B2021} \\ \hline
		$[16,8,5]_2$ & $[22,13,5]_2$ & $[38,21,5]_2$ & $[38,21,4]_2$  \\ 
		
		$[19,10,5]_2$ & $[20,11,5]_2$ & $[39,21,5]_2$ & $[39,21,4]_2$ \\ 
		
		$[17,9,5]_2$ & $[22,13,5]_2$ & $[39,22,5]_2$ & $[39,22,4]_2$  \\ 
		
		$[17,9,5]_2$ & $[23,12,6]_2$ & $[40,21,5]_2$ & $[40,21,4]_2$ \\ 
		
		$[20,11,5]_2$ & $[20,11,5]_2$ & $[40,22,5]_2$ & $[40,22,4]_2$ \\ 
		\hline
	\end{tabular}
\end{table}

\begin{remark}\label{rem.bounds on LCD}
	Some larger improvements of lower bounds on the minimum distance of binary 
	Euclidean LCD codes listed in \cite[Table 2]{B2021} were made in the recent article 
	\cite{IS2022} and the most recent preprint \cite{LS2022}. However,  note that our 
	method holds for any $q$, then on one hand, the results in Table \ref{tablast} indicate 
	that our method is simple and potentially effective in improving the lower bounds on 
	the minimum distance of LCD codes as well. On the other hand, as stated at the 
	beginning of this subsection, starting from these LCD codes with larger minimum 
	distances, more optimal LCD codes might be obtained by using the methods proposed in 
	the previous literature. 
\end{remark}

\section{Concluding remarks}\label{sec5}

In this paper, we consider two classical propagation rules, namely, the direct sum construction 
and the $(\mathbf{u},\mathbf{u+v})$-construction, and fully or partly determine their Euclidean and 
Hermitian hull dimensions in Theorems \ref{th.propagation1.direct sum}, \ref{th.propagation1.(u,u+v)-construction} 
and \ref{th.propagation2.(u,u+v)-construction}, respectively. 
As stated by Remarks \ref{rem1.direct sum}-\ref{rem5.mechanism of the two propagation rules}, 
our conclusions complement and generalize the previous results in the literature. 
In combination with Example \ref{exam.mechanism of the propagation rules}, one can 
easily find that our results are potentially effective for constructing new linear 
codes with prescribed hull dimensions. 
{As applications, we have constructed new 
	binary, ternary Euclidean FSD LCD codes, quaternary Hermitian LCD codes as well as Euclidean 
	and Hermitian (almost) self-orthogonal codes including (almost) self-dual codes in Theorems 
	\ref{th.2-ary FSD LCD codes}-\ref{th.4-ary FSD LCD codes}, \ref{th.binary sd with d=4}, 
	\ref{th.binary so with d=2^{t+2}}, \ref{th.binary so with d=2^t}, \ref{th.binary almost so with d=2^k-1}, 
	\ref{th.quaternary and octet Euclidean SO codes} and \ref{th.Hermitian SO codes over F4, F16, F64}.  
	In details, in Examples \ref{exam.binary sd with d=4}, \ref{exam.binary so with d=2^{t+2}}, 
	\ref{exam.binary so with d=2^t}, \ref{exam.binary almost so with d=2^t-1} and 
	Table \ref{tab:quaternary and octet Euclidean self-orthogonal codes}, on the basis of 
	our findings, we list numbers of Euclidean and Hermitian (almost) self-orthogonal codes, which 
	contain many (almost) optimal linear codes.}
Additionally, examples listed in Table \ref{tablast} show that our method is also potentially 
effective in improving the lower bounds on the minimum distance of LCD codes. 

\iffalse
Notice that we also propose some criteria for judging whether resulting codes are self-dual 
or self-orthogonal codes in Corollaries \ref{coro1.th.propagation1} and \ref{coro1.th.propagation2} 
under different conditions. In view of the fact that good binary Euclidean self-dual and self-orthogonal 
codes can be obtained from Theorems \ref{th.binary sd with d=4}, \ref{th.binary so with d=2^{t+2}}, 
\ref{th.binary so with d=2^t} and \ref{th.binary almost so with d=2^k-1}, we believe that good 
$q$-ary Euclidean and Hermitian self-dual and self-orthogonal codes can also be constructed 
similarly via Corollaries \ref{coro1.th.propagation1} and \ref{coro1.th.propagation2}. 
\fi
{Note that we have employed nested GRS codes to obtain some interesting Euclidean and Hermitian 
	self-orthogonal codes over finite fields with even characteristic, and some (almost) optimal 
	codes are contained. Hence, it would be interesting to use other types of linear codes and 
	consider finite fields with odd characteristic so that more (almost) optimal self-orthogonal codes 
	can be derived. Subsequently, good quantum codes can be further obtained from the well-known 
	CSS construction \cite{CSS1,CSS2} and Hermitian construction \cite{AKKS2007,KKK2006}. }

\ifCLASSOPTIONcaptionsoff
  \newpage
\fi



%\bibliographystyle{IEEEtran}

\begin{thebibliography}{99}
	\bibitem{AK1990} Assmus Jr, E.F., Key, J.D.: Affine and projective planes. Discrete Math. 83(2-3), 161-187 (1990).
	\bibitem{AK2008} Aly, S.A., Klappenecker, A.: Subsystem code constructions. In: IEEE Int. Symp. Inf. Theory (ISIT), 369-373 (2008). 
	%\bibitem{AH2019} Araya, M., Harada, M.: On the classification of linear complementary dual codes. Discrete Math. 342(1), 270-278 (2019).
	\bibitem{AH2020} Araya, M., Harada, M.: On the minimum weights of binary linear complementary dual codes. Cryptogr. Commun. 12(2), 285-300 (2020). 
	\bibitem{AHS2020} Araya, M., Harada, M., Saito, K.: Quaternary Hermitian linear complementary dual codes. IEEE Trans. Inf. Theory 66(5), 2751-2759 (2020).
	\bibitem{AHS2021.1} Araya, M., Harada, M., Saito, K.: On the minimum weights of binary LCD codes and ternary LCD codes. Finite Fields Appl. 76, 101925 (2021).
	%\bibitem{AHS2021.2} Araya, M., Harada, M., Saito K.: Characterization and classification of optimal LCD codes. Des. Codes Cryptogr. 89, 617-640 (2021). 
	\bibitem{AKKS2007} Aly, S.A., Klappenecker, A., Kumar, S., Sarvepalli, P.K.: On quantum and classical BCH codes. IEEE Trans. Inf. Theory 53(3), 1183-1188 (2007).
	
	\bibitem{B2021} Bouyuklieva, S.: Optimal binary LCD codes. Des. Codes Cryptogr. 89(11), 2445-2461 (2021).
	\bibitem{BCP1997} Bosma, W., Cannon, J., Playoust, C.: The Magma algebra system I: The user language. J. Symbolic Comput. 24(3-4), 235-265 (1997). 
	%\bibitem{BDH2006} Brun, T., Devetak, I., Hsieh, M.-H.: Correcting quantum errors with entanglement. Science. 314(5798), 436-439 (2006).
	\bibitem{BDHO1999} Bannai, E., Dougherty, S.T., Harada, M., Oura, M.: Type II codes, even unimodular lattices, and invariant rings. IEEE Trans. Inf. Theory. 45(4), 1194-1205 (1999).
	\bibitem{BN2001} Blackmore, T., Norton, G.H.: Matrix-product codes over $\F_q$. Appl. Algebra Eng. Commun. Comput. 12(6), 477-500 (2001).
	%\bibitem{C2022} Chen, H., New MDS entanglementassisted quantum codes from MDS Hermitian self-orthogonal codes. arXiv preprint arXiv:2206.13995, (2022). 
	\bibitem{CG2016} Carlet, C., Guilley, S.: Complementary dual codes for counter-measures to side-channel attacks. Adv. Math. Commun. 10(1), 131-150 (2016).
	\bibitem{CMTQ2018} Carlet, C., Mesnager, S., Tang, C., Qi, Y., Pellikaan, R.: Linear codes over $\F_q$ are equivalent to LCD codes for $q>3$. IEEE Trans. Inf. Theory 64(4), 3010-3017 (2018).
	%\bibitem{CY2022} Cao, M., Yang, J.: Intersections of linear codes and related MDS codes with new Galois hulls. arXiv preprint arXiv:2210.05551, (2022)
	\bibitem{CSS1} Calderbank, A.R., Shor, P.W.: Good quantum error-correcting codes exist. Phys. Rev. A 54(2), 1098-1105 (1996). 
	\bibitem{CS1998} Conway, J.H., Sloane, N.J.A.: Sphere Packings, Lattices and Groups, 3rd edn. Springer, New York (1998). 
	\bibitem{CXY2010} Chee, Y.M., Xing, C., Yeo, S.L.: New Constant-Weight Codes From Propagation Rules. IEEE Trans. Inf. Theory 56(4), 1596-1599 (2010).
	\bibitem{D2009} Ding, Y.: Asymptotic bound on binary self-orthogonal codes. Sci. China Series A: Math. 52(4), 631638 (2009). 
	%\bibitem{DKOSS2022} Dougherty, S.T., Kim, J.-L., Ozkaya, B., Sok, L., Solé, P.: The combinatorics of LCD codes: linear programming bound and orthogonal matrices. Int. J. Inf. Coding Theory 4, 116-128 (2017). 
	%\bibitem{DKSU2022} Dougherty, S.T., Korban, A., Șahinkaya, S., Ustun, D.: Binary self-dual and LCD codes from generator matrices constructed from two group ring elements by a heuristic search scheme. Adv. Math. Commun. (2022). 
	\bibitem{FFLZ2019} Fang, W., Fu, F.W., Li, L., Zhu, S.: Euclidean and Hermitian hulls of MDS codes and their applications to EAQECCs. IEEE Trans. Inf. Theory 66(6), 3527-3537 (2019).
	\bibitem{FZ2017} Fan, Y., Zhang, L.: Galois self-dual constacyclic codes. Des. Codes Cryptogr. 84(3), 473-492 (2017).
	\bibitem{GGJT2020} Guenda, K., Gulliver, T.A., Jitman, S., Thipworawimon, S.: Linear $\ell$-intersection pairs of codes and their applications. Des. Codes Cryptogr. 88(1), 133-152 (2020). 
	%\bibitem{GHMR2019} Galindo, C., Hernando, F.,  Matsumoto, R., Ruano, D.: Entanglement-assisted quantum error-correcting codes over arbitrary finite fields. Quantum Inf. Process. 18(4), 1-18 (2019).
	\bibitem{GJG2018} Guenda, K., Jitman, S., Gulliver, T.A.: Constructions of good entanglement-assisted quantum error correcting codes. Des. Codes Cryptogr. 86(1), 121-136 (2018).
	\bibitem{GKLRW2018} Galvez L., Kim J.L., Lee N., Roe Y.G., Won B.S.: Some bounds on binary LCD codes. Cryptogr. Commun. 10(4), 719-728 (2018).
	\bibitem{GR2015} Grassl, M., Rötteler, M.: Quantum MDS codes over small fields. In: IEEE Int. Symp. Inf. Theory (ISIT), 1104-1108 (2015).
	{\bibitem{Grass} Grassl, M.:  Bounds on the minimum distance of linear codes and quantum codes.
		Online available at \textrm{http://www.codetables.de}.
		Accessed on 2022-10-20.}
	\bibitem{H2007} Hou, X.D.: On the number of inequivalent binary self-orthogonal codes. IEEE Trans. Inf. Theory 53(7), 2459-2479 (2007). 
	%\bibitem{HS2019} Harada, M., Saito, K.: Binary linear complementary dual codes. Cryptogr. Commun. 11, 677-696 (2019). 
	\bibitem{HKS2021} Huffman, W.C., Kim, J.L., Solé, P.: Concise Encyclopedia of Coding Theory. Chapman and Hall/CRC (2021).
	\bibitem{I2022} Ishizuka, K.: Construction of quaternary Hermitian LCD codes. Cryptogr. Commun. (2022). https://doi.org/10.1007/s12095-022-00614-2. 
	\bibitem{IS2022} Ishizuka, K., Saito, K.: Construction for both self-dual codes and LCD codes. Adv. Math. Commun. (2022). https://doi.org/10.3934/amc.2021070. 
	\bibitem{JM2017} Jitman, S., Mankean, T.: Matrix-product constructions for Hermitian self-orthogonal codes. arXiv preprint arXiv:1710.04999, (2017).
	\bibitem{KKK2006} Ketkar, A., Klappenecker, A., Kumar, S.: Nonbinary stabilizer codes over finite fields. IEEE Trans. Inf. Theory 52(11), 4892-4914 (2006).
	\bibitem{KO2022} Kim, J.-L., Ohk, D.E.: DNA codes over two noncommutative rings of order four. J. Appl. Math. Comput. 68(3), 2015-2038 (2022). 
	\bibitem{KP1992} Kschischang, F.R., Pasupathy, S.: Some ternary and quaternary codes and associated sphere packings. IEEE Trans. Inf. Theory 38(2), 227-246 (1992).
	\bibitem{L1982} Leon, J.: Computing automorphism groups of errorcorrecting codes. IEEE Trans. Inf. Theory 28(3), 496-511 (1982).
	\bibitem{L1991} Leon, J.S.: Permutation group algorithms based on partition, I: Theory and algorithms. J. Symb. Comput. 12(4-5), 533-583 (1991).
	\bibitem{L2022} Li, S.: An improved method for constructing linear codes with small hulls. arXiv preprint arXiv:2204.04584, (2022). 
	
	\bibitem{LCC2018} Luo, G., Cao, X., Chen, X.: MDS codes with hulls of arbitrary dimensions and their quantum error correction. IEEE Trans. Inf. Theory 65(5), 2944-2952 (2018).
	\bibitem{LEGL2022} Luo, G., Ezerman, M.F., Grassl, M., Ling, S., How Much Entanglement Does a Quantum Code Need?. arXiv preprint arXiv:2207.05647, (2022).
	\bibitem{LEL2022} Luo, G., Ezerman, M.F., Ling, S.: Entanglement-assisted and subsystem quantum codes: New propagation rules and constructions. arXiv preprint arXiv:2206.09782, (2022).
	\bibitem{LEL2022LRC} Luo, G., Ezerman, M.F., Ling, S.: Three New Constructions of Optimal Locally Repairable Codes from Matrix-Product Codes. IEEE Trans. Inf. Theory (2022). https://doi.org/10.1109/TIT.2022.3203591.
	%\bibitem{LLFS20222022} Liu, Y., Li, R., Fu, Q., Song, H.: Minimum distances of binary optimal LCD codes of dimension five are completely determined. arXiv preprint arXiv:2210.05238. (2022). 
	%\bibitem{LLGF2015} Lu, L., Li, R., Guo, L., Fu, Q.: Maximal entanglement entanglement-assisted quantum codes constructed from linear codes. Quantum Inf. Process. 14(1), 165-182 (2015).
	\bibitem{LLY2020} Liu, X., Liu, H., Yu, L.: New EAQEC codes constructed from Galois LCD codes. Quantum Inf. Process. 19(1), 1-15 (2020).
	\bibitem{LLY2022} Liu, X., Liu, H., Yu, L.: New binary and ternary LCD codes from matrix-product codes. Linear Multilinear A. 70(5), 809-823 (2022).
	\bibitem{LMWX2022} Liu, S., Ma, L., Wu, T., Xing, C.: Good locally repairable codes via propagation rules. arXiv preprint arXiv:2208.04484, (2022).
	\bibitem{LP2020} Liu, H., Pan, X.: Galois hulls of linear codes over finite fields. Des. Codes Cryptogr. 88(2), 241-255 (2020).
	\bibitem{LS2022} Li, S., Shi, M.: Improved lower and upper bounds for LCD codes?. arXiv preprint arXiv:2206.04936, (2022).
	\bibitem{LX2004} Lin, S., Xing, C.: Coding Theory: A First Course. Cambridge, U.K.: Cambridge Univ. Press. (2004).
	\bibitem{LZYC2020} Lu, L., Zhan, X., Yang, S., Cao, H.: Optimal quaternary Hermitian LCD codes. arXiv preprint arXiv:2010.10166. (2020). 
	\bibitem{Edgar2004} Mart\'inez-Moro, E.: A generalization of Niederreiter-Xing's propagation rule and its commutativity with duality. IEEE Trans. Inf. Theory 50(4), 701-702 (2004).

	\bibitem{M1992} Massey, J.L.: Linear codes with complementary duals. Discrete Math. 106, 337-342 (1992).
	%\bibitem{M2022} Markus, G.: Bounds on the minimum distance of linear codes and quantum codes. Online available at http://www.codetables.de. Accessed on 2022-10-20.
	\bibitem{MCAK2021} Mahdjoubi, R., Cayrel, P.L., Akleylek, S., Kenza, G. A Novel Niederreiter-like cryptosystem based on the $(\mathbf{u}|\mathbf{u+v})$-construction codes. Rairo-theor. Inf. Appl. 55(10) (2021).
	\bibitem{MJ2016} Mankean, T., Jitman, S.: Matrix-product constructions for self-orthogonal linear codes. In 2016 12th International Conference on Mathematics, Statistics, and Their Applications (ICMSA) pp. 6-10, (2016).
	\bibitem{MPS2022} Meneghetti, A., Pellegrini, M., Sala, M.: A formula on the weight distribution of linear codes with applications to AMDS codes. Finite Fields Appl. 77, 101933 (2022).
	\bibitem{MS1977} MacWilliams, F.J., Sloane, N.J.A.: The theory of Error-Correcting codes. The Netherlands (1977). 
	\bibitem{MST1972} MacWilliams, F.J., Sloane, N.J.A., Thompson, J.G.: Good self-dual codes exist. Dis. Math. 3(1-3), 153-162 (1972).
	\bibitem{MVL2022} Miloslavskaya, V., Vucetic, B., Li, Y.: Computing the Partial Weight Distribution of Punctured, Shortened, Precoded Polar Codes. IEEE Trans. Commun. (2022). https://doi.org/10.1109/TCOMM.2022.3205967.
	
	\bibitem{NRS2006} Nebe, G., Rains, E.M., Sloane, N.J.A.: Self-Dual Codes and Invariant Theory, vol. 17. Springer, Berlin (2006). 
	
	
	\bibitem{CSS2} Steane, A.M.: Simple quantum error correcting codes. Phys. Rev. A 54(6), 4741-4751 (1996). 
	\bibitem{S2000} Sendrier, N.: Finding the permutation between equivalent codes: The support splitting algorithm. IEEE Trans. Inf. Theory 46(4), 1193-1203 (2000).
	\bibitem{SLK2022} Shi, M., Liu, N., Kim, J.L.: Classification of binary self-orthogonal codes of lengths from 16 to 20 and its application. J. Appl. Math. Comput. 1-31 (2022).
	\bibitem{SXS2020} Shi, M., Xu, L., Solé, P.: Construction of isodual codes from polycirculant matrices. Des. Codes Cryptogr. 88(12), 2547-2560 (2020). 
	%\bibitem{SXS2021} Shi, M., Xu, L., Solé, P.: On isodual double Toeplitz codes. arXiv preprint arXiv:2102.09233, (2021). 
	
	
	%\bibitem{TBTM2022} Talbi, S., Batoul, A., Tabue, A.F., Martínez-Moro, E.: Hulls of cyclic serial codes over a finite chain ring. Finite Fields Appl. 77, 101950 (2022). 
	
	
	\bibitem{ZLLL2020} Zhan, X., Li, R., Lu, L., Li, H.: Quatemary Hermitian linear complementary dual codes with small distance. In: 2020 International Conference on Information Science and Education (ICISE-IE). 38-41 (2020). 
	\bibitem{ZSO2022} Zhu, H., Shi, M., Özbudak, F.: Complete b-symbol weight distribution of some irreducible cyclic codes. Des. Codes Cryptogr. 90(5), 1113-1125 (2022). 
	
\end{thebibliography}


\end{document}